\documentclass[12pt]{iopart}

\usepackage[english]{babel}
\usepackage{graphicx}
\usepackage{textcomp}
\usepackage{iopams}
\usepackage{setstack}
\usepackage{amssymb}

\newtheorem{proposition}{Proposition}

\newtheorem{theorem}{Theorem}
\newenvironment{proof}[1][Proof:]{\begin{trivlist}
\item[\hskip \labelsep {\bfseries #1}]}{\end{trivlist}}

\newcommand{\id}{{\bf 1}}

\newcommand{\Spec}{\mbox{Spec}}

\newcommand{\EssSpec}{\mbox{EssSpec}}

\newcommand{\supp}{\mbox{supp}}

\newcommand{\Aut}{\mbox{Aut}}

\newcommand{\Ad}{\mbox{Ad}}

\newcommand{\ad}{\mbox{ad}}

\newcommand{\sgn}{\mbox{sgn}}

\newcount\MaxMatrixCols

\renewenvironment{matrix}{%
  \hskip -\arraycolsep\array{*\MaxMatrixCols c}%
}{%
  \endarray \hskip -\arraycolsep
}

\renewenvironment{pmatrix}{\left(\matrix}{\endmatrix\right)}

\begin{document}

\title[Harmonic analysis in Bianchi I-VII]{Explicit harmonic and spectral analysis in Bianchi I-VII type cosmologies}

\author{Zhirayr Avetisyan$^{1,2}$ and Rainer Verch$^2$}

\address{$^1$ Max Planck Institute for Mathematics in the Sciences, Inselstr. 22,\\
04103 Leipzig, Germany}

\address{$^2$ Institut f\"ur Theoretische Physik, Universit\"at Leipzig, Br\"uderstr. 16,\\
04103 Leipzig, Germany}

\ead{jirayrag@gmail.com}

\begin{abstract}
The solvable Bianchi I-VII groups which arise as homogeneity
groups in cosmological models are analyzed in a uniform manner.
The dual spaces (the equivalence classes of unitary irreducible
representations) of these groups are computed explicitly. It is
shown how parameterizations of the dual spaces can be chosen to
obtain explicit Plancherel formulas. The Laplace operator $\Delta$
arising from an arbitrary left invariant Riemannian metric on the
group is considered, and its spectrum and eigenfunctions are given
explicitly in terms of that metric. The spectral Fourier transform
is given by means of the eigenfunction expansion of $\Delta$. The
adjoint action of the group automorphisms on the dual spaces is
considered. It is shown that Bianchi I-VII type cosmological
spacetimes are well suited for mode decomposition. The example of
the mode decomposed Klein-Gordon field on these spacetimes is
demonstrated as an application.
\end{abstract}


\submitto{\CQG}

\maketitle

\section{Introduction}

While the cosmic microwave background radiation (CMB) has been
observed to be homogeneous and isotropic to a high degree of
accuracy,
data of the COBE and WMAP space missions
showed several interesting patterns (e.g., the
hemispherical asymmetry, the 'axes of evil', the cold spot etc.) that
could be taken as an indication of a slight anisotropy.
This has recently provided motivation to consider spatially homogeneous but
anisotropic cosmological spacetime models.
Because the anisotropy
is relatively small, only spacetime models which are close to
the standard, spatially isotropic Friedmann-Robertson-Walker (FRW)
geometries are compatible with oberservations.
Candidates are cosmological spacetimes with homogeneity groups
of Bianchi types I,V,VII or IX.
Attempts were made to fit such
Bianchi models to data so that the observed patterns are explained
\cite{Jaffe2005}, \cite{Jaffe2006}, \cite{Jaffe2006a}. Further
improvement included a cosmological constant, but all these
models were eventually ruled out \cite{Jaffe2006b},
\cite{Bridges2007}. Another possibility is suggested by the so
called early time Bianchi models (as opposed to late time Bianchi
models mentioned before).
In these models, the
the universe is considered as anisotropic at the very early epoch before and
during inflation, and then goes over into
the spatially isotropic FRW geometry
 \cite{Dechant2009}, \cite{Soda2012},\cite{Pontzen2011}.
The newly acquired data from the Planck mission seem to be compatible with
this scenario \cite{Collaboration2013},
\cite{Collaboration2013a}.

But anisotropic cosmological spacetime models are also of
mathematical interest. Mathematically, it is very natural to
investigate such models with their associated isometry groups from
the perspective of harmonic analysis in homogeneous spaces. Unlike
FRW models, these models are yet to be investigated in more
detail. Another field of interest is mathematical cosmology, which
as yet lacks a thorough study aside from FRW models. Of particular
interest are different scenarios of the early epoch of the
universe including also quantum effects. The adequate theoretical
framework for describing such situations is quantum field theory
in curved spacetimes (QFT in CST; see \cite{Wald1994} as a general
reference; see also \cite{Verch2012}, \cite{Brunetti2005},
\cite{Hollands2010} and references cited therein for more recent
developments). In the context of (anisotropic) cosmological
models, harmonic analysis and Fourier analysis with respect to the
isometry groups of the spacetimes under consideration are powerful
tools for obtaining rigorous and explicit results within the
setting of QFT in CST.

There is a prevalent tendency among different scientific
disciplines to specialize towards particular aims and perspectives and
thus to diverge from each other. One may view the
mathematical theory of harmonic analysis as an example of that
phenomenon. The Levi decomposition
effectively breaks apart the general harmonic analysis to those
for solvable and semisimple groups separately. Both branches have
been investigated in great generality. The Kirillov orbit theory
for solvable groups, along with such developments as the Currey
theory for exponential solvable Lie groups, has given methods for
obtaining explicit results for arbitrary dimensions. On the other
hand the theory developed by Harish-Chandra, Helgason and others
provides a very deep insight into general semisimple homogeneous
spaces and related structures. At the same time, abstract harmonic
analysis for locally compact groups, and the Mackey machine based
upon it, yielded remarkably in describing general principles and
phenomena. However, when a theoretical cosmologist wants to
perform a mode decomposition of some physical field on a
homogeneous spacetime it is almost of no use to him to know that,
for instance, a Borel measure exists, or that it can be computed
up to equivalence for $N$ dimensions. What a cosmologist really
needs is an explicit description with formulas which can be used
without expert knowledge in harmonic analysis, as it is available
for the traditional Abelian group $\mathbb{R}^n$. But as far as we
were able to see, results of that kind do not seem to be available
in the literature; possibly because they are usually not in the
focus of interest of mathematicians. It may be regarded as one of
the tasks of mathematical physics to provide bridges between
occasionally diverging interests of mathematics and physics. With
this in mind we take to the task of giving an explicit description
of harmonic analysis of a number of groups of cosmological
interest.

As stated above, in quantum field theory on cosmological
spacetimes one is often interested in Fourier analysis with
respect to the symmetries of the underlying geometry. In the
framework of homogeneous cosmologies, $M=\mathcal{I}\times\Sigma$
where $\mathcal{I}$ is an open interval, and $\Sigma$ is a smooth
three dimensional manifold with a Riemannian metric $h$ on it. The
spacetime metric $g$ on $M$ takes the line element form
$$
ds_g^2=dt^2-a_{ij}(t)ds_h^ids_h^j
$$
with a positive definite matrix $a_{ij}(t)$ depending smoothly on
the "time coordinate" $t\in\mathcal{I}$. $(\Sigma,h)$ is a
homogeneous space with respect to a Lie group of isometries $G$.
The variety of all such homogeneous spaces which arise in
cosmology can be found in several works including
\cite{Petrov59},\cite{StephaniKramerMacCallumHoenselaersHerlt200305},\cite{pittphilsci1507},\cite{Ryan1975}\cite{Ellis}.
We have analyzed the geometrical setup of general linear
hyperbolic fields on cosmological spacetimes in
\cite{ZhAPub12012}. In particular we have seen that in nearly all
cases one deals with a semidirect product group
$G=\Sigma\rtimes\mathcal{O}$ where $\mathcal{O}$ is either of
$SO(3)$, $SO(2)$, $\{1\}$, and the role of $\Sigma$ is played by
Bianchi I-IX groups $Bi(N)$ ($N=$I,II,...,IX) and their quotients
$Bi(N)/\Gamma$ by discrete normal subgroups $\Gamma$. Such
structures are called {\it semidirect homogeneous spaces}, and a
few important results have been obtained in \cite{ZhAPub12012} in
this generality using abstract harmonic analysis. We moreover have
seen how tightly harmonic and spectral analysis is related to mode
decomposition. The next step towards physics will be to describe
harmonic analysis of all those possible semidirect homogeneous
spaces explicitly. The spaces of maximal symmetry with
$\mathcal{O}=SO(3)$ are the FRW spaces, which are described by
isometry groups $SO(4)$, $E(3)$ or $SO^+(1,3)$. The spaces with
one rotational symmetry are described by $\mathcal{O}=SO(2)$ and
are called LRS (locally rotationally symmetric) spaces. And
finally the purely homogeneous spaces are given by trivial
isotropy groups $\mathcal{O}=\{1\}$. The isometry groups of FRW
spaces are classical groups and their harmonic analysis is also a
classical subject. For purely homogeneous spaces (otherwise called
Bianchi spaces) this is known partially. The Bianchi I group is
the additive group $\mathbb{R}^3$ of which harmonic analysis is
textbook standard. The Bianchi II group is the famous Heisenberg
group of dimension $2+1$, which is well studied, and its harmonic
analysis can be found in
\cite{Thangavelu199803},\cite{Strichartz89},\cite{Folland199502}.
The Bianchi III group is the $ax+b$ group in $2+1$ dimensions
whose harmonic analysis is known as well \cite{Folland199502}. The
Bianchi VIII group is the universal covering group
$\widetilde{SL(2,\mathbb{R})}$, which is a notorious non-linear
group. Its harmonic analysis can be found in \cite{Pukanszky64}.
The last group, Bianchi IX, is simply $SU(2)$ which is again
classical, and its harmonic analysis can be found, e.g., in
\cite{Folland199502}. Little is known about the Bianchi IV-VII
groups beyond the structure of their Lie algebras which are
semidirect products of Abelian algebras $\mathbb{R}^2$ and
$\mathbb{R}$. In fact, although there are principally no obstacles
on the way of their investigation, we were not able to locate any
explicit description of their harmonic analysis in the literature;
we also asked some prominent experts in the field, and none was
able to point to such references. Even less is known about the
semidirect products of Bianchi groups with $SO(2)$ describing LRS
models. We stress again that there is no obstacle to applying the
Mackey machine and perform all calculations, but it seems that
this has not been done so far. Because the discrete subgroups
$\Gamma$ can be readily found from the group structure
\cite{Osinovsky1973}, once having control over harmonic analysis
of the group $G$ it is not hard to reduce it to the quotient
$G/\Gamma$, but it again needs to be done somewhere. We have
chosen to start with harmonic analysis of the solvable Bianchi
I-VII groups in a uniform manner deferring the remaining
structures to the future. One could argue that this might be done
straightforwardly by the exponential solvable methods of
\cite{Currey2005}, but if one actually starts to do that (what we
indeed did) one needs to construct an enormous amount of spaces,
dual spaces and intersections, which are designed to handle
arbitrary groups, and seem to be too bulky to be performed by
hand. Therefore we preferred the original Mackey construction.

The harmonic analytical Fourier transform is only half the way to
applications. In fact what is really useful in concrete
computations is the spectral Fourier transform given by the
eigenfunction expansion of a left invariant Laplace operator
$\Delta$ acting on the sections in the subbundle of $\mathcal{T}$
over $\Sigma$. In \cite{ZhAPub12012} we have seen how the mode
decomposition of linear fields can be performed knowing the
eigenfunction expansion of $\Delta$. We moreover have seen how in
principle one can identify the harmonic analytical Fourier
transform with the spectral Fourier transform of $\Delta$ to
translate the general results on semidirect spaces into the
language of concrete calculations. For this purpose one needs to
find the spectral theory of $\Delta$ explicitly for all possible
$\Sigma$ and invariant Riemannian metrics on them. This is again
something that can hardly be found in the literature, although
spectral analysis in Riemannian spaces is huge and very well
developed a subject in mathematics. In particular, one needs to
know the spectrum and a complete system of eigenfunctions of
$\Delta$ explicitly. In general, the eigenfunction problem of
$\Delta$ is a vector valued elliptic partial differential equation
on a manifold without boundary, which is difficult to compute even
numerically. If this equation admits separation of variables so
that the eigenfunctions are given by combinations of functions of
one variable subject to ordinary differential equations then we
can consider these eigenfunctions as given explicitly in terms of
special functions. In this work we will give such an explicit
description of the spectrum and eigenfunctions of $\Delta$ in
terms of an arbitrary left invariant Riemannian metric on $\Sigma$
for the line bundle over Bianchi I-VII groups. For arbitrary
bundle dimension this is much more complicated. There is a bit of
hope to obtain explicit solutions by transforming the original
vector valued eigenfunction equation on the manifold to a scalar
elliptic eigenfunction equation with constraints on the holonomy
bundle of the linear connection associated with the given fiber
metric. This is a non-trivial task to which we hope to return in
the future. The remaining Bianchi VIII and IX groups are
semisimple, and the spectral analysis on them goes beyond the
methods of the present paper.

We summarize the content of the our exposition as follows.
First the Bianchi I-VII groups are realized in Section
\ref{Section_Semi} as semidirect products of $\mathbb{R}^2$ and
$\mathbb{R}$ and the main group properties are explicitly
computed, such as the multiplication laws, exponential maps,
modular functions and adjoint representations. Then in Section
\ref{Section_Irrep} the dual spaces of the groups are constructed,
i.e., the equivalence classes of unitary irreducible
representations. This is done by means of the Mackey machine. Next
in Section \ref{Section_Orbits} a look is given at the co-adjoint
orbits of the groups in the sense of the Kirillov theory, and it
is described explicitly how the cross sections can be chosen to
parameterize the dual space. Afterwards, in Section
\ref{Section_Plancherel}, an explicit Plancherel formula is given
for all these groups. Thereafter we turn to spectral analysis. The
spectra and the eigenfunctions of $\Delta$ are found explicitly in
Section \ref{SecSpecScal} in terms of the chosen arbitrary left
invariant Riemannian metric. Then in Section \ref{Section_Fourier}
it is shown that these eigenfunctions are complete in $L^2$ and
give rise to a conventional Fourier transform in sense of
\cite{ZhAPub12012}. In the final part the applications in quantum
field theory are discussed. First in Section
\ref{Section_SepVarHomUni} it is shown that Bianchi spacetimes are
ideally adapted for the mode decomposition as given in
\cite{ZhAPub12012}. Then in Section \ref{Section_KG} this mode
decomposition is demonstrated in the example of the Klein-Gordon
field on Bianchi I-VII spacetimes. A number of interesting
consequences are indicated including those for the quantum
Klein-Gordon field, where results from \cite{ZhAThesis} are used
as well.

\section{Semidirect structure of Bianchi I-VII groups\label{Section_Semi}}

As a first step we will try to explicitly realize the solvable
Bianchi II-VII groups  (I is Abelian and will serve as a starting
point in the analysis of others) as semidirect products of Abelian
subgroups. A classification of solvable real Lie algebras with
respect to such products can be inferred from
\cite{Mubarakzyanov63}.

{\bf Semidirect products of Lie algebras and Lie groups.} We start
by recalling some definitions. Let $\mathfrak{a}$ and
$\mathfrak{b}$ be Lie algebras, and let $D(\mathfrak{a})$ be the
Lie algebra of derivations on $\mathfrak{a}$. Let further
$f:\mathfrak{b}\mapsto D(\mathfrak{a})$ be a Lie algebra
homomorphism. The {\it semidirect product Lie algebra}
$\mathfrak{a}\times_f\mathfrak{b}$ is the algebra modelled on
$\mathfrak{a}\oplus\mathfrak{b}$ with the Lie bracket
$$
[(a,b),(a',b')]=([a,a']+f(b)a'-f(b')a,[b,b'])\mbox{,
}(a,b),(a',b')\in\mathfrak{a}\oplus\mathfrak{b}.
$$
Let, on the other hand, $A$ and $B$ be Lie groups, and
$F:B\mapsto\Aut(A)$ a Lie group homomorphism ($\Aut(A)$ embedded
into $GL(A)$). The {\it semidirect product} $A\times_F B$ of
groups $A$ and $B$ is defined as the Lie group modelled on the
product manifold $A\times B$ with the multiplication
$$
(a,b)(a',b')=(aF(b)a',bb'),\qquad (a,b),(a',b')\in A\times_F B.
$$
Following the notation of \cite{Hochschild1965}, denote by
$F^\circ:B\mapsto\Aut(\mathfrak{a})$ the map $B\ni b\mapsto
d[F(b)]\in\Aut(\mathfrak{a})$, where $\mathfrak{a}$ is the Lie
algebra of $A$. Then the derivative of this map, $f=dF^\circ$,
will be a Lie algebra homomorphism $f:\mathfrak{b}\mapsto
D(\mathfrak{a})$ ($\mathfrak{b}$ the Lie algebra of $B$), and the
Lie algebra of the direct product Lie group $A\times_F B$ is the
direct product Lie algebra $\mathfrak{a}\times_f\mathfrak{b}$
\cite{Hochschild1965}.

{\bf Bianchi I-VII groups as semidirect products.} With this in
mind, let us start with realizing Bianchi I-VII algebras as
semidirect product algebras
$\mathfrak{g}=\mathbb{R}^2\times_f\mathbb{R}$ with some Lie
algebra homomorphism $f:\mathbb{R}\mapsto D(\mathbb{R}^2)$. This
correspondence between Bianchi algebras and homomorphisms $f$ can
be obtained by combination of \cite{McCallum1998} and
\cite{Mubarakzyanov63}. (Those uncomfortable with Russian may
simply perform the semidirect product construction and check the
commutation relations.) Namely, in each case $f(r)=r\cdot M$,
$r\in\mathbb{R}$, in a suitable basis, where $M$ is a $2\times2$
matrix. The matrix $M$ for each algebra is given in Table
\ref{MTable} below.

\begin{table}
\centering
\begin{tabular}{|c|c|c|c|c|c|c|}
\hline

I & II & III & IV & V & VI & VII\\\hline

0 &

$\begin{pmatrix} 0 & 0\\
1 & 0
\end{pmatrix}$ &

$\begin{pmatrix}
1 & 0\\
0 & 0
\end{pmatrix}$ &

$\begin{pmatrix} 1 & 0\\
1 & 1
\end{pmatrix}$ &

$\begin{pmatrix} 1 & 0\\
0 & 1
\end{pmatrix}$ &

$\begin{pmatrix} 1 & 0\\
0 & -q
\end{pmatrix}$ &

$\begin{pmatrix} p & -1\\
1 & p
\end{pmatrix}$\\

 & & & & & $-1<q\le1$ & $p\ge0$\\\hline
\end{tabular}
\caption{\label{MTable}The matrices $M$ for Bianchi I-VII groups}
\end{table}

The corresponding integral homomorphisms $F^\circ$ will be the
exponentials $F^\circ(r)=e^{rM}$ (note that the exponential map on
the group $\mathbb{R}$ is given by the identity map). If a
diffeomorphism is given locally by a linear coordinate map,
$x_i'=A_i^jx_j$ with the matrix $A$, then its differential will be
given by the same matrix $A$. Now that $F^\circ(r)=d[F(r)]$ and
that $F(r)$ are linear automorphisms, it follows that
$F(r)=e^{rM}$. Thus all Bianchi groups I-VII are given by
semidirect products $G=\mathbb{R}^2\times_F\mathbb{R}$, where for
each class the group homomorphism
$F:\mathbb{R}\mapsto\Aut(\mathbb{R}^2)$ is given as in Table
\ref{F_rTable} above.

\begin{table}
\centering
\begin{tabular}{|c|c|c|c|c|c|c|}
\hline

I & II & III & IV & V & VI & VII\\\hline

1 &

$\begin{pmatrix} 1 & 0\\
r & 1
\end{pmatrix}$ &

$\begin{pmatrix}
e^r & 0\\
0 & 1
\end{pmatrix}$ &

$\begin{pmatrix} e^r & 0\\
re^r & e^r
\end{pmatrix}$ &

$\begin{pmatrix} e^r & 0\\
0 & e^r
\end{pmatrix}$ &

$\begin{pmatrix} e^r & 0\\
0 & e^{-qr}
\end{pmatrix}$ &

$e^{pr}\begin{pmatrix} \cos(r) & -\sin(r)\\
\sin(r) & \cos(r)
\end{pmatrix}$\\

 & & & & & $-1<q\le1$ & $p\ge0$\\\hline
\end{tabular}
\caption{\label{F_rTable}The matrices $F(r)$ for Bianchi I-VII
groups}
\end{table}

We appoint to use capital symbols $X,Y,Z$ for Lie algebra
coordinates and small symbols $x,y,z$ for Lie group coordinates,
but these may interfere in some calculations involving exponential
maps. It follows that the group multiplication is
$$
(x,y,z)(x',y',z')=((x,y)+F(z)(x',y'),z+z')\mbox{,
}(x,y,z),(x',y',z')\in G=\mathbb{R}^2\times_F\mathbb{R}.
$$

{\bf The exponential map.} Finally we note that all 7 groups are
exponential, and the exponential map is given as follows. Let
$(X,Y,Z)\in\mathfrak{g}=\mathbb{R}^2\times_f\mathbb{R}$ with
$(X,Y)\in\mathbb{R}^2$ and $Z\in\mathbb{R}$. We use the Zassenhaus
formula
$$
\exp(A+B)=\exp(A)\exp(B)\exp(C_2)\exp(C_3)...,
$$
where the coefficients $C_m$ are homogeneous Lie algebra elements
composed of nested commutators of order $m$. We will use the
convenient method of obtaining $C_m$ recursively as given in
\cite{Magnus1954}. If we set $A=(X,Y,0)$ and $B=(0,0,Z)$, we
obtain
$$
[A,B]=-f(Z)A.
$$
Now equating the homogeneous summands of any order of (4.7) and
(4.8) of \cite{Magnus1954}, we obtain recursion formulas for $C_m$
which are bulky in general. However, trying an ansatz
$C_m=\alpha_m (-f)^{m-1}(Z)A$, $\alpha_m\in\mathbb{R}$, and
checking directly for $m=2$, one can easily prove it inductively,
and find
$$
\alpha_m=\frac{1-m}{m!}.
$$
It remains to calculate
$$
\exp(C_2)\exp(C_3)...=\exp\left(\sum_{m=1}^\infty
\frac{1-m}{m!}(-f)^{m-1}(Z)A\right).
$$
If $f(Z)$ is invertible for all $Z$ then we write
$$
\frac{1-m}{m!}(-f)^{m-1}(Z)=(-f)^{-1}(Z)\frac{(-f)^m(Z)}{m!}-\frac{(-f)^{m-1}(Z)}{(m-1)!},
$$
and obtain
\begin{eqnarray}
D(Z)\doteq\sum_{m=1}^\infty
\frac{1-m}{m!}(-f)^{m-1}(Z)=(-f)^{-1}(Z)\left(e^{-f(Z)}-1\right)-e^{-f(Z)}\nonumber\\
=f^{-1}(Z)\left(1-F(-Z)\right)-F(-Z).\label{D_ZDef}
\end{eqnarray}
It is only for Bianchi II and III that $f(Z)$ is degenerate, and
for these two we can compute directly
$$
D(Z)=\sum_{m=1}^\infty
\frac{1-m}{m!}(-f)^{m-1}(Z)=\frac{1}{2}f(Z)\qquad\mbox{for Bianchi
II}
$$
and
$$
D(Z)=\sum_{m=1}^\infty
\frac{1-m}{m!}(-f)^{m-1}(Z)=(1-2e^{-1})f(Z)\qquad\mbox{for Bianchi
III.}
$$
Thus we arrive at
$$
\exp((X,Y,0)+(0,0,Z))=\exp((X,Y,0))\exp((0,0,Z))\exp(D(Z)(X,Y),0).
$$
The exponential maps of $\mathbb{R}^2$ and $\mathbb{R}$ are the
identity maps, therefore
$$
(x,y,z)=\exp((X,Y,Z))=(X,Y,Z)(D(Z)(X,Y),0)=([1+F(Z)D(Z)](X,Y),Z),
$$
where $F(Z)$ should be understood as $F(\exp(Z))$. The matrices
$D(Z)$ appear somewhat bulky so we refrain from presenting them in
a table.

{\bf The adjoint representations $\Ad$ and $\ad$.} Let
$(g_x,g_y,g_z),(x,y,z)\in G$. Their conjugation
$(x',y',z')=(g_x,g_y,g_z)(x,y,z)(g_x,g_y,g_z)^{-1}$ is given by
$$
(x',y',z')=((1-F(z))(g_x,g_y)+F(g_z)(x,y),Z).
$$
The adjoint representation $\Ad$ is the differential of this map
at the identity element $(x,y,z)=(0,0,0)$, and so it is given by
the matrix field $\Ad_g$,
$$
\Ad_g=
\begin{pmatrix}
F(g_z) & -F'(0)\begin{pmatrix}g_x\\ g_y\end{pmatrix}\\
\begin{matrix}0 &  0\end{matrix} & 1
\end{pmatrix}.
$$
The adjoint representation of the Lie algebra is given by the
matrix\footnote{We use the general relation $\ad_{\sf X}${\sf
Y}=[{\sf X},{\sf Y}] for elements ${\sf X}$, ${\sf Y}$ in a
general Lie algebra.}
$$
\ad_{(X,Y,Z)}=
\begin{pmatrix}
f(Z) & -f'(0)\begin{pmatrix}X\\ Y\end{pmatrix}\\
\begin{matrix}0 & 0\end{matrix} & 0
\end{pmatrix}.
$$

{\bf The Haar measure and the modular function.} The Haar measure
on the Lie group is given by
$$
dg=d(\exp(X,Y,Z))=j(X,Y,Z)dXdYdZ,
$$
where
$$
j(X,Y,Z)=\mathfrak{h}\det\frac{1-e^{-\ad_{(X,Y,Z)}}}{\ad_{(X,Y,Z)}}\mbox{,
}(X,Y,Z)\in\mathfrak{g},
$$
and $0<\mathfrak{h}\in\mathbb{R}$ is an arbitrary constant. In
group coordinates one can check that $dg$ is given by
$$
dg=\mathfrak{h}\det F(-g_z)dg_zdg_ydg_z\mbox{, }(g_x,g_y,g_z)\in
G.
$$
The groups are all non-compact, so there is no preferred
normalization for the constant $\mathfrak{h}$. Later it will be
determined as related to the chosen left invariant Riemannian
metric on $G$. The modular function $\Delta(g)=\det\Ad_g^{-1}$ can
be readily seen to be $\Delta(g)=\det F(-g_z)$.

This temporarily completes our task of analyzing the Bianchi I-VII
groups as semidirect products. In the next section we will
concentrate on their dual spaces.

\section{The irreducible representations of Bianchi I-VII groups \label{Section_Irrep}}

In this section we will try to find the dual spaces of Bianchi
I-VII groups using the Mackey procedure. Let us start with Bianchi
I, which is simply the additive group $\mathbb{R}^3$. Its dual
group $\hat{\mathbb{R}}^3$ is homeomorphic to itself,
$\hat{\mathbb{R}}^3=\mathbb{R}^3$, and the irreducible
1-dimensional representations are given by
$$
\xi_{\vec k}(\vec x)=e^{i\{\vec k,\vec x\}}\mbox{, }\vec
x\in\mathbb{R}^3\mbox{, }\vec k\in\hat{\mathbb{R}}^3=\mathbb{R}^3,
$$
where we appoint to denote by $\{\vec a,\vec b\}$ the usual
Euclidean product of three-vectors $\vec a,\vec b\in\mathbb{R}^3$.
These scalar functions $\xi_{\vec k}$ can be viewed as unitary
operator valued functions acting on the one complex dimensional
Hilbert space $\mathbb{C}$.

{\bf The Mackey procedure for normal Abelian subgroups.} We cite
here the setup of the Mackey theory for groups with a normal
Abelian subgroup as given in \cite{Folland199502}. Let $G$ be a
locally compact group and $N$ an Abelian normal subgroup. Then $G$
acts on $N$ by conjugation, and this induces an action of $G$ on
the dual group $\hat N$ defined by
$$
g\nu(n)=\nu(g^{-1}ng)\mbox{, }g\in G\mbox{, }\nu\in\hat N\mbox{,
}n\in N.
$$
For each $\nu\in\hat N$, we denote by $G_\nu$ the {\it stabilizer}
of $\nu$,
$$
G_\nu=\{g\in G\mbox{: }g\nu=\nu\},
$$
which is a closed subgroup of $G$, and we denote by
$\mathcal{O}_\nu$ the {\it orbit} of $\nu$:
$$
\mathcal{O}_\nu=\{g\nu\mbox{: }g\in G\}.
$$
The action of $G$ on $\hat N$ is said to be {\it regular} if some
conditions are satisfied. To avoid presenting excessive
information we only mention that if $G$ is second countable (which
is true for a Lie group), then the condition for a regular action
is equivalent to the following: for each $\nu\in\hat N$, the
natural map $gG_\nu\mapsto g\nu$ from $G/G_\nu$ to
$\mathcal{O}_\nu$ is a homeomorphism. In our case $\hat N$ is a
smooth manifold, and the group actions are all smooth, hence this
map is not only a homeomorphism but even a diffeomorphism. The
constructions become simpler under the assumption that $G$ is a
semidirect product of $N$ and the factor group $H=G/N$. We define
the {\it little group} $H_\nu$ of $\nu\in\hat N$ to be
$H_\nu=G_\nu\cap H$. Now we cite a beautiful theorem which appears
as {\bf Theorem 6.42} in \cite{Folland199502} and expresses the
essence of the Mackey procedure. The notation $\mbox{\bf Ind}$ and
the inducing construction are briefly introduced in the {\bf
Appendix A}.

\begin{theorem} [Folland,\cite{Folland199502}] Suppose $G=N\ltimes H$, where $N$ is Abelian and
$G$ acts regularly on $\hat N$. If $\nu\in\hat N$ and $\rho$ is an
irreducible representation of $H_\nu$, then $\mbox{\bf
Ind}^G_{G_\nu}(\nu\rho)$ is an irreducible representation of $G$,
and every irreducible representation of $G$ is equivalent to one
of this form. Moreover, $\mbox{\bf Ind}^G_{G_\nu}(\nu\rho)$ and
$\mbox{\bf Ind}^G_{G_\nu}(\nu'\rho')$ are equivalent if and only
if $\nu$ and $\nu'$ belong to the same orbit, say $\nu'=g\nu$, and
$h\mapsto\rho(h)$ and $h\mapsto\rho'(g^{-1}hg)$ are equivalent
representations of $H_\nu$.
\end{theorem}

{\bf Application to the Bianchi groups.} It is easy to see that
Bianchi groups II-VII satisfy the assumptions of the theorem. In
this case $N=\mathbb{R}^2$ and $H=\mathbb{R}$, the dual of $N$ is
$\hat N=\mathbb{R}^2$ and is given by
$$
\hat N=\{e^{i\{\breve k,\breve x\}}\mbox{: }\breve x,\breve
k\in\mathbb{R}^2\},
$$
where we overload the notation by brackets $\{\breve a,\breve b\}$
to denote the two dimensional Euclidean product of $\breve
a,\breve b\in\mathbb{R}^2$. Let $\imath_N:\mathbb{R}^2\mapsto G$
be the natural inclusion. The action of $G$ on $\hat N$ is given
by
$$
g\xi_{\breve k}(\breve x)=\xi_{\breve
k}(\imath_N^{-1}(g^{-1}\imath_N(\breve x)g)).
$$
All Bianchi solvable groups are homeomorphic to $\mathbb{R}^3$,
and we may choose a global chart on them. In particular we choose
one adapted to the semidirect structure
$\mathbb{R}^2\times_{F}\mathbb{R}$ presented in the previous
section. Then the multiplication law in $G$ is given by
$$
(x,y,z)(x',y',z')=((x,y)+F(z)(x',y'),z+z').
$$
The unit $e\in G$ is given by $e=(0,0,0)$, and the inverse map by
$$
(x,y,z)^{-1}=(-F^{-1}(z)(x,y),-z).
$$
In particular, if $(\breve x,0)=(x,y,0)\in\imath_N(\mathbb{R}^2)$
and $(g_x,g_y,g_z)\in G$, then
$$
(g_x,g_y,g_z)^{-1}(x,y,0)(g_x,g_y,g_z)=(F^{-1}(g_z)(x,y),0),
$$
that is, the conjugation map $n\mapsto g^{-1}ng$ is given by
$(x,y)\mapsto F^{-1}(g_z)(x,y)$. Thus the action of $G$ on $\hat
N$ is
$$
g\xi_{\breve k}(\breve x)=\xi_{\breve k}(F^{-1}(g_z)\breve
x)=e^{i\{\breve k,F^{-1}(g_z)\breve x\}}=e^{i\{F^\bot(g_z)\breve
k,\breve x\}},
$$
where $F^\bot(g_z)$ is the inverse transpose of the matrix
$F(g_z)$. This means that this action can be described by
$$
g\breve k=F^\bot(g_z)\breve k\mbox{, }g\in G\mbox{, }\breve
k\in\mathbb{R}^2.
$$
Denote by $V^0\subset\mathbb{R}^2$ the eigenspace of $M^\top$
corresponding to the eigenvalue $0$ (the null space). Then it will
be also the joint eigenspace of the matrices
$F^\bot(g_z)=e^{-g_zM^\top}$ corresponding to the eigenvalue $1$
simultaneously for all $g_z\in\mathbb{R}$. Let us write the
stabilizer condition,
$$
e^{-g_zM^\top}\breve k=\breve k.
$$
Then the stabilizer $G_{\breve k}$ and the little group $H_{\breve
k}$ will be
$$
G_{\breve k}=\imath_N(\mathbb{R}^2)\cdot H_{\breve k}
$$
and
$$
H_{\breve k}=\cases{\mathbb{R} & if $\breve k\in V^0$,\\
\{0\} & else. }
$$
Define the following space of irreducible representations of $G$:
$$
\hat J=(V^0\times\mathbb{R})\cup(\mathbb{R}^2\setminus V^0).
$$
For each $\mu\in\hat J$ the corresponding irreducible
representation is given by
$$
T_\mu(g)=e^{i\{\breve k,\breve g\}}e^{ik_3g_3}\mbox{, }\mu=(\breve
k,k_3)=(k_1,k_2,k_3)
$$
if $\mu\in V^0\times\mathbb{R}$, and
$$
T_\mu=T_{\breve k}=\mbox{\bf Ind}^G_{\mathbb{R}^2}(e^{i\{\breve
k,.\}})\mbox{, }\mu=\breve k,
$$
if $\mu\in\mathbb{R}^2\setminus V^0$. The orbit
$\mathcal{O}_{\breve k}$ is $\{\breve k\}$ if $\breve k\in V^0$
and $F^\bot(\mathbb{R})\breve k$ otherwise. As mentioned in the
theorem, two representations $\mu,\mu'\in\hat J$ are equivalent if
and only if $\breve k$ and $\breve k'$ are on the same orbit,
$\breve k=F^\bot(z)\breve k'$, and the corresponding
representations of $H_{\breve k}$ and $H_{\breve k'}$ are
equivalent when intertwined with the action of $z$. The first
condition can be satisfied non-trivially if $\breve k,\breve
k'\in\mathbb{R}^2\setminus V^0$, but then $H_{\breve k}=H_{\breve
k'}=\{0\}$, and thus there exists only the trivial representation
$\rho=1$. Thus representations $\mu,\mu'\in\mathbb{R}^2\setminus
V^0$ are equivalent if and only if they are on the same orbit. On
the other hand, let $\mu,\mu'\in V^0\times\mathbb{R}$ such that
$\breve k=\breve k'$, and the first condition is satisfied
trivially. Then $G_{\breve k}=G$, and $G/G_{\breve k}=\{1\}$, so
the action of $1$ cannot intertwine inequivalent representations
of $H_{\breve k}$. Thus $\mu\sim\mu'$ means $\mu=\mu'$. Therefore
the dual space $\hat G$ of $G$ will be
$$
\hat G=(V^0\times\mathbb{R})\cup(\mathbb{R}^2\setminus
V^0)/F^\bot(\mathbb{R}).
$$

{\bf The null spaces $V^0$.} Finally let us find the eigenspaces
$V^0$ for different Bianchi groups. By a calculation of
eigenvectors and eigenvalues of $M$ we obtain
$$
V_I^0=\mathbb{R}^2\mbox{, }V_{II}^0=\mathbb{R}\oplus\{0\}\mbox{,
}V_{III}^0=\{0\}\oplus\mathbb{R},
$$
$$
V_{IV}^0=\{0\}\mbox{, }V_V^0=\{0\}\mbox{, }V_{VII}=\{0\},
$$
and
$$
V_{VI}= \cases{
\{0\}\oplus\mathbb{R} & if $q=0$, \\
\{0\} & else.}
$$
As it is already visible from Table \ref{MTable}, the Bianchi
group VI with $q=0$ coincides with Bianchi III group, and this
fact is also reflected in the null spaces above. Note that as
always with solvable groups, the irreducible representations are
either 1-dimensional or infinite dimensional.

To obtain explicit descriptions of the dual groups $\hat G$ for
each Bianchi class we have to calculate the orbits
$\mathcal{O}_{\breve k}=F^\bot(\mathbb{R})\breve k$ explicitly,
which is done in the next section. Note that the entire
construction could have been performed through the machinery of
exponential solvable Lie groups developed in \cite{Currey1991} and
\cite{Currey2005}, where the problem is treated exhaustively. In
particular, it was shown that (as adapted to our terminology)
there exists a cross section $\tilde K$, an algebraic submanifold
of $\mathbb{R}^2$ which crosses each {\it generic orbit} (i.e., an
orbit of maximal dimension) exactly once, and thus parameterizes
the infinite dimensional representations. Having explicitly
calculated $\tilde K$ we find $\hat
G=(V^0\times\mathbb{R})\cup\tilde K$. But the methods of
\cite{Currey1991} are extremely general and involve simple but
lengthy algebraic calculations; this is why we have preferred the
original topological Mackey constructions.

\section{Co-adjoint orbits of Bianchi II-VII groups \label{Section_Orbits}}

The term co-adjoint orbits would probably suit better to the
solvable Lie theoretical method of orbits as established by
Kirillov and accomplished by Currey. At this point we deviate to
present a little digression demonstrating the equivalence of that
approach with that we have adopted.

{\bf The Kirillov approach.} The Lie algebra
$\mathfrak{g}=\mathbb{R}^2\times_f\mathbb{R}$ of $G$ is modelled
on the vector space $\mathbb{R}^3$, and as such its dual space
$\mathfrak{g}'$ is again isomorphic to $\mathbb{R}^3$. We will fix
this isomorphism by choosing the basis in $\mathfrak{g}'$ dual to
our adapted basis of $\mathfrak{g}$. With this identification the
co-adjoint action of $G$ on $\mathfrak{g}'=\mathbb{R}^3$ is given
by the matrix field $\Ad_g^*=\Ad_g^\bot$,
$$
\Ad_g^*=
\begin{pmatrix}
F^\bot(g_z) & & 0\\
       & & 0\\
(g_x,g_y)F^\bot(g_z)M^\top & & 1
\end{pmatrix}.
$$
For any $\mathfrak{l}=(X^*,Y^*,Z^*)\in\mathfrak{g}'$ its orbit
$\mathcal{O}_\mathfrak{l}$ is given by
$$
\mathcal{O}_\mathfrak{l}=(F^\bot(\mathbb{R})(X^*,Y^*),(\mathbb{R},\mathbb{R})F^\bot(\mathbb{R})M^\top(X^*,Y^*)+Z^*),
$$
and the space of orbits $\{\mathcal{O}_\mathfrak{l}\}$ with the
quotient topology induced from $\mathfrak{g}'$ is homeomorphic to
$\hat G$ with the Fell topology \cite{Folland199502}. One can
easily see that the orbits are of two types: those of
$(X^*,Y^*,Z^*)$ with $(X^*,Y^*)\in V^0$ or $(X^*,Y^*)\notin V^0$.
The former are the so called {\it degenerate} orbits with
dimension $0$ (singletons), and the latter are the {\it generic}
orbits with maximal dimension $3$. This is exactly the same result
we obtained above by Mackey machine.\index{Representation,
generic, singleton}

{\bf The generic orbits and the cross sections.} We will denote
the range of a parameterized quantity $Q(p)$ of a parameter $p\in
P$ by $Q(P)$. For instance, $F^{\bot}(\mathbb{R})$ will denote the
range of the quantity $F^{\bot}(r)$ when $r$ runs over
$\mathbb{R}$. Here we will try to find the generic orbits
$F^\bot(\mathbb{R})\breve k_0\in\hat G$ mentioned in the previous
section and corresponding cross-sections $\tilde
K\in\mathbb{R}^2$. The latter will be algebraic manifolds composed
of one or more connected components. In all cases $V^0$ is a
subset of Lebesgue measure $0$ in $\mathbb{R}^2$. By the
definition of the cross section $\tilde K$, the submanifold
$\mathbb{R}^2\setminus V^0$ can be parameterized by a global chart
$\breve k=\breve k(k,r)$, $(k,r)\in\mathfrak{K}\times\mathbb{R}$,
such that $\breve k(k,r)=F^{\bot}(r)\breve k_0(k)$ and $\breve
k_0(k)=\breve k(k,0)\in\tilde K$. Under this diffeomorphism the
Lebesgue measure $d\breve k$ becomes $\rho(k,r)dkdr$, where
$\rho(k,r)=|\det\partial(\breve k)/\partial(k,r)|$.

Now let us proceed to the determination of the orbits and the
cross sections case by case. Figure 1 in {\bf Appendix B}
illustrates them qualitatively.

II. We have
$$
F^{\bot}(r)(k_x,k_y)=(k_x-rk_y,k_y),
$$
hence the orbit through $\breve k\in\mathbb{R}^2\setminus V^0$ is
$F^{\bot}(\mathbb{R})(k_x,k_y)=(\mathbb{R},k_y)$. The cross
section can be chosen to be $\tilde K=\breve k_0(\mathfrak{K})$,
$\mathfrak{K}=\mathbb{R}\setminus\{0\}$, $\breve k_0(k)=(0,k)$.
Indeed, any orbit $(\mathbb{R},k_y)$ meets $\tilde K$ exactly once
at $\breve k_0(k_y)$. Then
$$
\rho(k,r)=\left|\det\left(F^{\bot}(r)\frac{\partial\breve
k_0(k)}{\partial k},\frac{\partial F^{\bot}(r)}{\partial r}\breve
k_0(k)\right)\right|=|k|.
$$

III. In this case
$$
F^{\bot}(r)(k_x,k_y)=(e^{-r}k_x,k_y),
$$
and the orbit through $\breve k\in\mathbb{R}^2\setminus V^0$ is
$F^{\bot}(\mathbb{R})(k_x,k_y)=(\sgn(k_x)\mathbb{R}_+,k_y)$. Let
$\mathfrak{K}=\mathbb{R}\times\{-1,1\}$, $k=(k_1,k_2)$. The cross
section is the image $(-1,\mathbb{R})\cup(1,\mathbb{R})$ of the
map $\breve k_0(k)=(k_2,k_1)$. We find
$$
\rho(k,r)=e^{-r}.
$$

IV. For this group
$$
F^{\bot}(r)(k_x,k_y)=(e^{-r}k_x-re^{-1}k_y,e^{-r}k_y),
$$
and the orbits are complicated. We set
$\mathfrak{K}=\mathbb{R}_{+0}\times\{-1,1\}$, $k=(k_1,k_2)$ and
$\breve k_0(k)=(k_2,k_2k_1)$. That this is a cross section can be
checked immediately. The measure density $\rho$ is
$$
\rho(k,r)=e^{-2r}(1+k_1).
$$

V. Now
$$
F^{\bot}(r)(k_x,k_y)=e^{-r}(k_x,k_y),
$$
and the orbits are simply the incoming radial rays. Set
$\mathfrak{K}=\mathbb{R}/2\pi\mathbb{Z}$ and $\breve
k_0(k)=(\cos(k),\sin(k))$. It follows that
$$
\rho(k,r)=e^{-2r}.
$$

VI. For this group we consider only the case $q\neq 0$ as $q=0$ is
just the group III.
$$
F^{\bot}(r)(k_x,k_y)=(e^{-r}k_x,e^{qr}k_y),
$$
and the orbits are incoming polynomial curves if $q<0$ and
hyperbolic curves if $q>0$. For $q<0$ set
$\mathfrak{K}=\mathbb{R}/2\pi\mathbb{Z}$ and $\breve
k_0(k)=(\cos(k),\sin(k))$. Then
$$
\rho(k,r)=e^{-(1-q)r}(\cos^2(k)-q\sin^2(k)).
$$
For $q>0$ set $\mathfrak{K}=\mathbb{R}_{+0}\times\{0,1,2,3\}$,
$k=(k_1,k_2)$ and
$$
\breve k_0(k)=
\begin{pmatrix}
0 & -1\\
1 & 0
\end{pmatrix}^{k_2}
\begin{pmatrix}
1\\
k_1
\end{pmatrix}.
$$
Thus
$$
\rho(k,r)=q^{k_2\bmod2}e^{-(1-q)r}.
$$

VII. The co-adjoint action in this group is given by
$$
F^{\bot}(r)(k_x,k_y)=e^{-pr}(k_x\cos r-k_y\sin r,k_x\sin r+k_y\cos
r),
$$
and the orbits are incoming or outgoing spirals depending on
whether $p<0$ or $p>0$. We take $\mathfrak{K}=(-e^{\pi
p};-1]\cup[1,e^{\pi p})$ and $\breve k_0(k)=(k,0)$. Each orbit
clearly intersects $\tilde K$ exactly once. Finally
$$
\rho(k,r)=e^{-2pr}|k|.
$$

Note that in all cases we have chosen $\tilde K$ such that it
possesses an involution $\breve k_0(-k)=-\breve k_0(k)$, which
will be useful in later constructions. Of course, these choices of
cross sections are not unique, neither need they correspond to
those suggested by Currey theory. In fact, one may make any other
choice for convenience and calculate the corresponding measure
density $\rho$ precisely as we did.

\section{The explicit Plancherel formula for Bianchi II-VII groups \label{Section_Plancherel}}

We will obtain the Plancherel measure by extending the idea
suggested in \cite{Folland199502} for Heisenberg groups to all
solvable Bianchi groups. Namely, we will exploit the Euclidean
Parseval equality on the homeomorphic space $\mathbb{R}^3$.

{\bf Introductory material.} Before going to the solvable groups
II-VII let us recall the well-known form of the Plancherel formula
for the Abelian group $\mathbb{R}^3$. The Fourier transform of a
function $f\in C^\infty_0(\mathbb{R}^3)$ is defined by
$$
\hat f(\vec k)=\int_{\mathbb{R}^3}d\vec xe^{-i\{\vec k,\vec
x\}}f(\vec x),
$$
and the Plancherel formula is
$$
\int_{\mathbb{R}^3}d\vec x|f(\vec
x)|^2=(2\pi)^3\int_{\mathbb{R}^3}d\vec k|\hat f(\vec k)|^2.
$$
The Plancherel measure is simply $d\nu(\vec k)=(2\pi)^3 d\vec k$,
proportional to the Lebesgue measure on $\mathbb{R}^3$.

We start by noting that, being an algebraic (matrix) group, $G$ is
necessarily type I (Theorem 7.8 or 7.10 \cite{Folland199502}), and
the normal subgroup $N$ is unimodular and therefore in the kernel
of the modular function $\Delta$. It follows from {\bf Theorem
7.6} in \cite{Folland199502} that the Mackey Borel structure on
$\hat G$ is standard, and thereby due to {\bf Lemma 7.39} in
\cite{Folland199502} we have a measurable field of representations
$\pi_p$ on $p\in\hat G$, such that $\pi_p\in p$ (or equivalently,
we have a measurable choice of representatives of each equivalence
class $[\pi]\in\hat G$). Henceforth we will speak of a
representation $\pi\in\hat G$ meaning the value of this field at a
given point $[\pi]\in\hat G$. As can be inferred from
\cite{Currey2005} in the language of solvable Lie groups, only
those irreducible representations corresponding to the generic
orbits (i.e., orbits of maximal dimension) admit a non-zero
Plancherel measure. Therefore only $T_\mu$ with
$\mu\in\mathbb{R}^2\setminus V^0$ (generic representations) will
play a role in the Fourier transform. We proceed to their
construction as $T_{\breve k}=\mbox{\bf
Ind}^G_{\mathbb{R}^2}(e^{i\{\breve k,.\}})$ following \S 6.1 in
\cite{Folland199502}.

{\bf The Fourier transform at generic representations.} For each
$\breve k\in\mathbb{R}^2\setminus V^0$ the representation Hilbert
space $H_{\breve k}$ of $\nu_{\breve k}=e^{i\{\breve k,.\}}$ is
$H_{\breve k}=\mathbb{C}$. The homogeneous space $G/N=\mathbb{R}$
has a natural $G$-invariant measure, which is the Lebesgue measure
$dz$. The representation space of $T_\mu$ is then the completion
$L^2(\mathbb{R},\mathbb{C})$ of the space of compactly supported
continuous sections in the homogeneous Hermitian line bundle
$\mathbb{R}\times\mathbb{C}$, and the action of $G$ on it is given
by
$$
T_{\breve k}(g)f[z]=e^{-i\{\breve
k,(g^{-1}z)_N\}}f[(g^{-1}z)_H]=e^{i\{\breve k,F(-z)\breve
g\}}f[z-g_z]\mbox{, }g=(\breve g,g_z)\in G\mbox{, }f\in
C_0(\mathbb{R},\mathbb{C}),
$$
where for any $g\in G$ we write $g=g_Ng_H$, $g_N\in N$, $g_H\in
H$. For $f\in C_0(G)$ define the (harmonic analytical) Fourier
transform by
$$
\hat f(\pi)=\pi(f)D_\pi=\int_Gf(g)\pi(g)D_\pi dg,
$$
where the operator $D_\pi$ is defined on $\phi\in
C_0(\mathbb{R},\mathbb{C})$ by
$$
D_\pi\phi[z]=\Delta(z)^{+\frac{1}{2}}\phi[z]=(\det
F(z))^{-\frac{1}{2}}\phi[z].
$$
(Note that there is a misprint in the formula (7.49) of
\cite{Folland199502}, and the sign "$-$" in the power of $\Delta$
should be replaced by "$+$". The author of \cite{Folland199502}
confirmed this in a private communication.) By {\bf Theorem 7.50}
in \cite{Folland199502}, the operator fields $\hat f(\pi)$ are
measurable fields of Hilbert-Schmidt operators, and if we identify
the space of Hilbert-Schmidt operators on $H_\pi$ with the tensor
product space $H_\pi\otimes H_\pi^*$, then the Fourier transform
gives an isomorphism
$$
L^2(G)\sim\int_{\hat G}^\oplus d\nu(\pi)H_\pi\otimes H_\pi^*.
$$
To find the Plancherel measure $d\nu(\pi)$ we calculate the
Fourier transforms $\hat f(T_{\breve k})$ directly. For $\phi\in
C_0(\mathbb{R},\mathbb{C})$ we have
$$
\hat f(T_{\breve k})\phi[r]=\int_Gf(g)T_{\breve
k}(g)D_\pi\phi[r]dg
$$
$$
=\mathfrak{h}\int_{\mathbb{R}^3}dg_xdg_ydg_zf(g_x,g_y,g_z)e^{i\{\breve
k,F(-r)\breve g\}}(\det F(g_z))^{-1}(\det
F(r-g_z))^{-\frac{1}{2}}\phi[r-g_z]
$$
by a substitution $g_z'=r-g_z$
$$
=\mathfrak{h}\int_{\mathbb{R}^3}dg_xdg_ydg_z'f(g_x,g_y,r-g_z')e^{i\{\breve
k,F(-r)\breve g\}}(\det F(r-g_z'))^{-1}(\det
F(g_z'))^{-\frac{1}{2}}\phi[g_z']
$$
$$
=\mathfrak{h}\int_{\mathbb{R}^3}dg_xdg_ydg_z'f(g_x,g_y,r-g_z')e^{i\{F^\bot(r)\breve
k,\breve g\}}(\det F(r-g_z'))^{-1}(\det
F(g_z'))^{-\frac{1}{2}}\phi[g_z'].
$$
Thus $\hat f(T_{\breve k})$ is an integral operator with a smooth
kernel
$$
\mathcal{K}^f_{\breve
k}(r,g_z')=\mathfrak{h}\bar{\mathfrak{F}}_{\mathbb{R}^2}[f(.,.,r-g_z')](F^\bot(r)\breve
k)(\det F(r-g_z'))^{-1}(\det F(g_z'))^{-\frac{1}{2}},
$$
where
$$
\bar{\mathfrak{F}}_{\mathbb{R}^2}[\psi(.,.)](\breve
k)=\int_{\mathbb{R}^2}dxdy\psi(x,y)e^{i\{\breve k,\breve x\}}.
$$
The Hilbert-Schmidt norm $|\|\hat f(T_{\breve k})\||$ is given by
$$
|\|\hat f(T_{\breve
k})\||^2=\mathfrak{h}^2\int_{\mathbb{R}^2}drdg_z'|\mathcal{K}^f_{\breve
k}(r,g_z')|^2.
$$
Coming back to the original variable $g_z=r-g_z'$, we have
$$
|\|\hat f(T_{\breve
k})\||^2=\mathfrak{h}^2\int_{\mathbb{R}^2}drdg_z\left|\bar{\mathfrak{F}}_{\mathbb{R}^2}[f(.,.,g_z)](F^\bot(r)\breve
k)\right|^2(\det F(g_z))^{-2}(\det F(r-g_z))^{-1}
$$
$$
=\mathfrak{h}^2\int_{\mathbb{R}^2}drdg_z\left|\bar{\mathfrak{F}}_{\mathbb{R}^2}[f(.,.,g_z)](F^\bot(r)\breve
k)\right|^2(\det F(g_z))^{-1}(\det F(r))^{-1}
$$
(by Fubini's theorem)
$$
=\mathfrak{h}^2\int_\mathbb{R}dg_z(\det
F(g_z))^{-1}\int_\mathbb{R}dr\left|\bar{\mathfrak{F}}_{\mathbb{R}^2}[f(.,.,g_z)](F^\bot(r)\breve
k)\right|^2(\det F(r))^{-1}.
$$

{\bf The Plancherel formula.} Now we refer to the previous section
about the co-adjoint orbits, and note that in all cases
$\rho(k,r)=\dot\nu(k)(\det F(r))^{-1}$ with some continuous
non-negative function $\dot\nu(k)$ on $\tilde K$. We will shortly
see that
\begin{eqnarray}
d\nu(k)=\mathfrak{h}^{-1}\dot\nu(k)dk\label{dnu_k}
\end{eqnarray}
is exactly the Plancherel measure desired. Indeed,
$$
\int_{\tilde K}dk\mathfrak{h}^{-1}\dot\nu(k)|\|\hat f(T_{\breve
k_0(k)})\||^2=\mathfrak{h}\int_{\tilde
K}dk\dot\nu(k)\int_\mathbb{R}dg_z(\det F(g_z))^{-1}
$$
$$
\times\int_\mathbb{R}dr\left|\bar{\mathfrak{F}}_{\mathbb{R}^2}[f(.,.,g_z)](F^\bot(r)\breve
k_0(k))\right|^2(\det F(r))^{-1}
$$
by another application of Fubini's theorem (see
\cite{Dieudonne1976},chapter XIII),
$$
=\mathfrak{h}\int_\mathbb{R}dg_z(\det F(g_z))^{-1}\int_{\tilde
K}dk\int_\mathbb{R}dr\rho(k,r)\left|\bar{\mathfrak{F}}_{\mathbb{R}^2}[f(.,.,g_z)](F^\bot(r)\breve
k_0(k))\right|^2
$$
by definition of $\rho(k,r)$,
$$
=\mathfrak{h}\int_\mathbb{R}dg_z(\det
F(g_z))^{-1}\int_{\mathbb{R}^2}d\breve
k\left|\bar{\mathfrak{F}}_{\mathbb{R}^2}[f(.,.,g_z)](\breve
k)\right|^2
$$
by the Euclidean Parseval formula,
$$
=\mathfrak{h}\int_\mathbb{R}dg_z(\det
F(g_z))^{-1}\int_{\mathbb{R}^2}dg_xdg_y\left|f(g_x,g_y,g_z)\right|^2=\int_Gdg|f(g)|^2,
$$
thus we arrive at an explicit Plancherel forumla,
$$
\int_{\tilde K}d\nu(k)|\|\hat f(T_{\breve
k_0(k)})\||^2=\int_Gdg|f(g)|^2.
$$
The Plancherel measures for groups II-VII are thus given by
\begin{eqnarray}
\dot\nu_{II}(k)=|k|\mbox{, }\dot\nu_{III}(k)=1\mbox{,
}\dot\nu_{IV}(k)=1+k_1\mbox{, }\dot\nu_V(k)=1\mbox{,
}\dot\nu_{VI-}(k)=\cos^2(k)-q\sin^2(k)\nonumber\\
\dot\nu_{VI+}(q)=q^{k_2\bmod2}\mbox{,
}\dot\nu_{VII}(k)=|k|.\label{dotnu_k}
\end{eqnarray}
Note that we could have chosen the cross section for $VI$, $q<0$
in the same way as for $VI$, $q>0$ to get a uniform Plancherel
measure $\dot\nu_{VI}=\dot\nu_{VI_+}$ for all Bianchi VI groups,
but we preferred the more conventional circle to the quartet of
rays in Figure 1 when it was possible. This can be altered for any
technical purposes when needed.

\section{Scalar spectral analysis on Bianchi I-VII groups\label{SecSpecScal}}

Here the term scalar spectral analysis is understood as the
spectral theory of the scalar Laplacian. Of course, there is no
distinguished Laplacian on these groups. We will consider {\it any
Laplacian} which arises as the metric operator with respect to any
conserved metric on the group.\index{Laplace operator}

Let $G$ be one of these groups, and let $\mathfrak{L}(G)$ be its
Lie algebra generated by three right invariant vector fields
$\xi_1,\xi_2,\xi_3$. Let further $X_1,X_2,X_3$ be a basis of left
invariant vector fields on $G$, and
$d\omega^1,d\omega^2,d\omega^3$ the dual basis. Any left invariant
metric $h_{ab}$ on $G$ can be written as
$$
h_{ab}=\sum_{i,j=1}^3\check h_{ij}d\omega^i_ad\omega^j_b,
$$
where $\check h_{ij}$ is any symmetric positive definite
$3\times3$ matrix, and the corresponding metric Laplacian will be
\begin{eqnarray}
\mathbf{\Delta}_h=\sum_{i,j=1}^3\check
h^{ij}X_iX_j,\label{LaplaceDef}
\end{eqnarray}
with $\check h^{ij}=(\check h_{ij})^{-1}$. To see this first note
that
$$
\sum_{i,j=1}^3\check h^{ij}X_iX_jf=\sum_{i,j=1}^3\check
h^{ij}\sum_{l,m=1}^3[X_i^lX_j^m\partial_l\partial_m+X_i^l(\partial_lX_j^m)\partial_m]f.
$$
On the other hand the connection Laplacian related to the
Levi-Civita connection is given by
$$
\mathbf{\Delta}_h=\sum_{i,j=1}^3\check
h^{ij}\sum_{l,m=1}^3[X_i^lX_j^m\partial_l\partial_m-\sum_{k=1}^3X_i^lX_j^m{\bf\Gamma}^k_{lm}\partial_k],
$$
where the ${\bf\Gamma}^k_{lm}$ are the Christoffel symbols. This
together with the observation
$$
\sum_{l,m=1}^3X_i^lX_j^m{\bf\Gamma}^k_{lm}=-\frac{1}{2}(X_i^l\partial_lX_j^m+X_j^l\partial_lX_i^m),
$$
which follows from $\nabla_{X_i}X_j=\frac{1}{2}[X_i,X_j]$, gives
(\ref{LaplaceDef}). Our aim will be to find the eigenfunctions and
the spectrum of $\mathbf{\Delta}_h$. If $\xi_1$ and $\xi_2$
commute and also commute with all $X_i$, then
$\xi_1,\xi_2,\mathbf{\Delta}_h$ are a triple of commuting
operators, and have common eigenfunctions. We will find those
eigenfunctions and show that they are complete in the sense we
desire. For the ease of notation let us denote
$$
\check h^{2\times2}=\check h^{ij}|_{i,j<3},
$$
\begin{eqnarray}
\check h^{\bullet3}=\check h^{ij}|_{i<j=3},\qquad\check
h^{3\bullet}=\check h^{ij}|_{j<i=3}.\label{MatrixNot}
\end{eqnarray}

First let us describe the spectrum $\Spec(\mathbf{\Delta}_h)$ of
the Laplacian $\mathbf{\Delta}_h$. We note that
$\mathbf{\Delta}_h$ is a negative semidefinite operator, as
$(\mathbf{\Delta}_hf,f)_{L^2(G)}=-(d_hf,d_hf)_{L^2(G)}\le 0$,
where $d_h$ is the exterior derivative with respect to the metric
$h_{ab}$. Thus $\mathbf{\Delta}_h$ is a semibounded and real
symmetric operator on $L^2(G)$. There are several ways of
extending $\mathbf{\Delta}_h$ to a self-adjoint operator on
$L^2(G)$. A real symmetric operator has a self-adjoint extension
by von Neumann's theorem \cite{Reed1975}. A semibounded symmetric
operator has a self-adjoint extension by Friedrich's extension
theorem \cite{Reed1975}. But we have something stronger. The Lie
group $G$ with its left invariant Riemannian metric $h_{ab}$ is a
complete Riemannian manifold \cite{Milnor1976}. Then following
\cite{Chernoff1973} $\mathbf{\Delta}_h$ is essentially
self-adjoint on $C_0^\infty(G)$. Being a negative self-adjoint
operator $\mathbf{\Delta}_h$ has a real non-positive spectrum,
$\Spec(\mathbf{\Delta}_h)\subset(-\infty;0]$. The semidirect
structure of our groups satisfies the conditions of Lemma 5.6 of
\cite{Milnor1976}, and we have for the scalar curvature $R_h$ the
following formula,
\begin{eqnarray}
R_h=-Tr[S^2]-(Tr[S])^2,\label{R_h}
\end{eqnarray}
where we took into account that the normal Lie subgroup
$\mathbb{R}^2$ with the induced metric is flat. The matrix $S$ is
given by
$$
S=\frac{1}{2}(\ad_{(0,0,1)}|_{\mathbb{R}^2}+\ad_{(0,0,1)}|_{\mathbb{R}^2}^*)=\frac{1}{2}(f(1)+f(1)^*),
$$
where the adjoint $^*$ is understood as
$$
h(Af,g)=h(f,A^*g)\mbox{, }\forall f,g\in \mathfrak{L}(G)\mbox{,
}A\in\Aut(\mathfrak{L}(G)).
$$
Here $h(f,g)$ for $f,g\in\mathfrak{L}(G)$ means the evaluation of
the Riemannian metric $h$ on the vector fields $f$,$g$. Thus all
our groups endowed with any left invariant Riemannian metric are
spaces of constant negative curvature equal to $R_h$, which is
given explicitly in terms of the matrices $f(1)=M$ and $\check
h^{2\times2}$ with
\begin{eqnarray}
S=\frac{1}{2}(M+(\check h^{2\times2})^{-1}M\check
h^{2\times2}).\label{S_Mh}
\end{eqnarray}
This in turn implies, following \cite{Donnelly1981}, that the
essential spectrum of $\mathbf{\Delta}_h$ is precisely
$\EssSpec(\mathbf{\Delta}_h)=(-\infty;R_h]$. Recall that the
essential spectrum of a self adjoint operator consists of
eigenvalues of infinite multiplicity (see \cite{Reed1975}). For
the group Bianchi I, all irreps are 1-dimensional, and as we will
see later in the section, each eigenspace representation includes
an infinite number of them, thus there is no discrete spectrum.
For the remaining groups, we have seen in the previous section
that no finite dimensional representation enters the Plancherel
formula. On the other hand, in the next section we will see that
the infinite dimensional eigenspaces exhaust $L^2(G)$, hence no
finite dimensional eigenspace exists, i.e., the discrete spectrum
is empty, and therefore
$\Spec(\mathbf{\Delta}_h)=\EssSpec(\mathbf{\Delta}_h)$.

To find the generators $\xi_i$ for Bianchi I-VII groups we
differentiate the left translation map $\vec x\mapsto g\vec x$,
$$
g(x,y,z)=((g_x,g_y)+F(g_z)(x,y),g_z+z),
$$
and obtain
$$
\begin{pmatrix}
\xi_1\\
\xi_2\\
\xi_3
\end{pmatrix}=
\begin{pmatrix}
1 & 0 & 0\\
0 & 1 & 0\\
(x, & y)\dot F^\top(0)& 1
\end{pmatrix}
\begin{pmatrix}
\partial_x\\
\partial_y\\
\partial_z
\end{pmatrix}.
$$
We see that $\xi_1=\partial_x$ and $\xi_2=\partial_y$ do indeed
commute. To find the left invariant vectors $X_i$ (which are the
generators of right translations) we differentiate the right
translation map $\vec x\mapsto\vec xg$,
$$
(x,y,z)g=((x,y)+F(z)(g_x,g_y),z+g_z),
$$
and get
\begin{eqnarray}
\begin{pmatrix}
X_1\\
X_2\\
X_3
\end{pmatrix}=
\begin{pmatrix}
F^\top(z) & \begin{matrix}0 \\ 0\end{matrix}\\
\begin{matrix}0 & 0\end{matrix} & 1
\end{pmatrix}
\begin{pmatrix}
\partial_x\\
\partial_y\\
\partial_z
\end{pmatrix}.\label{X_iExplicit}
\end{eqnarray}
Thus $\xi_1$, $\xi_2$ do commute with all $X_i$. Now let
$\zeta(\vec x)\in C^\infty(G)$ be a joint eigenfunction for
$\{\xi_1,\xi_2,\mathbf{\Delta}_h\}$. Then it is necessarily of the
form
$$
\zeta(\vec x)=e^{i\{\breve k_\mathbb{C},\breve x\}}P(z),
$$
where $\breve k_\mathbb{C}\in\mathbb{C}^2$, $\breve x=(x,y)$, and
satisfies
$$
\mathbf{\Delta}_h\zeta(\vec x)=\lambda\zeta(\vec x),
$$
for some $\lambda\in\mathbb{C}$. A matrix representation of
equation (\ref{LaplaceDef}) and a bit of manipulation yields the
following equation
$$
\check h^{33}\ddot P(z)+i(\breve k_\mathbb{C}^\top F(z)[\check
h^{\bullet3}+(\check h^{3\bullet})^\top])\dot P(z)
$$
$$
-(\lambda+\breve k_\mathbb{C}^\top F(z)\check
h^{2\times2}F^\top(z)\breve k_\mathbb{C}-i\check
h^{3\bullet}F^\top(z)M^\top\breve k_\mathbb{C})P(z)=0,
$$
where $\dot F^\top(z)=\partial_z e^{zM^\top}=F^\top(z)M^\top$ was
used. This is a generalized time-dependent harmonic oscillator
equation, which always have solutions, and those solutions
comprise a two complex dimensional space. For given $\lambda$ and
$\breve k_\mathbb{C}$ let us choose two linearly independent
solutions $P_{\lambda,k_\mathbb{C}}(z)$ and
$Q_{\lambda,k_\mathbb{C}}(z)$ (the choice of initial data may be
arbitrary).

First we consider the group Bianchi I. Here $M=0$, $F(z)=1$ and
the equation becomes
$$
\check h^{33}\ddot P(z)+i(\breve k_\mathbb{C}^\top[\check
h^{\bullet3}+(\check h^{3\bullet})^\top])\dot P(z)-(\lambda+\breve
k_\mathbb{C}^\top\check h^{2\times2}\breve k_\mathbb{C})P(z)=0.
$$
One can easily check that $P(z)=e^{ik_z\cdot z}$ is a solution if
$$
\lambda=-\vec k_\mathbb{C}^\top\check h^{ij}\vec k_\mathbb{C},
$$
where $\vec k_\mathbb{C}=(\breve k_\mathbb{C},k_z)$. This is a
consequence of the fact that for this group $\xi_3$ also commutes
with all $\xi_i$ and $X_i$, so that there exist joint
eigenfunctions of the commuting operators
$\xi_1,\xi_2,\xi_3,\mathbf{\Delta}_h$ of the form
$$
\zeta(\vec x)=e^{i\{\vec k_\mathbb{C},\vec x\}},
$$
corresponding to the eigenvalues
$$
\lambda=-\vec k_\mathbb{C}^\top\check h^{ij}\vec k_\mathbb{C}.
$$
In particular, when we restrict ourselves to the irreducibles
$\vec k_\mathbb{C}=\vec k\in\mathbb{R}^3$, we obtain
$$
\zeta_{\vec k}(\vec x)=e^{i\{\vec k,\vec x\}},
$$
and we observe immediately that each eigenspace corresponding to
the eigenvalue $\lambda$ includes infinitely many $\vec k$ which
satisfy
$$
\lambda=-\vec k^\top\check h^{ij}\vec k.
$$
Of course, $e^{ik_z\cdot z}$ do not exhaust all solutions $P(z)$.
But it turns out that the $\zeta_{\vec k}$ constructed in this way
are already complete in $L^2(G)$. Indeed, that is the essence of
the Euclidean Parseval equality. To be more precise, we need to
take $d\nu(\vec k)=\frac{1}{\mathfrak{h}}d\vec k$ as the
Plancherel measure for the Euclidean Plancherel formula to hold.
Equivalently we can renormalize $\zeta_{\vec k}$ by taking
$$
\zeta_{\vec k}(\vec x)=\frac{1}{\sqrt{\mathfrak{h}}}e^{i\{\vec
k,\vec x\}}
$$
so that the Plancherel measure is independent of the metric. But
this ease of construction is a peculiarity which the remaining
groups Bianchi II-VII do not share, and we proceed to determine
their eigenfunctions.

For the groups II-VII let us now restrict to
$0>\lambda\in\mathbb{R}$ and $\breve
k_\mathbb{C}=F^\bot(r)k_0(-k)\in\mathbb{R}^2\setminus V^0$,
$k\in\mathfrak{K}$, $r\in\mathbb{R}$ (minus sign for convenience).
The equation now becomes
\begin{eqnarray}
\check h^{33}\ddot P(z)+i(\breve k_0(-k)^\top F(z-r)[\check
h^{\bullet3}+(\check h^{3\bullet})^\top])\dot
P(z)\nonumber\\
-(\lambda+\breve k_0(-k)^\top F(z-r)\check
h^{2\times2}F^\top(z-r)\breve k_0(-k)\nonumber\\
-i\check h^{3\bullet}F^\top(z-r)M^\top\breve
k_0(-k))P(z)=0,\label{EqPr}
\end{eqnarray}
and the two independent solutions will be denoted by
$P_{\lambda,k,r}(z)$ and $Q_{\lambda,k,r}(z)$. If we set
$P_{\lambda,k,0}(z)=P_{\lambda,k}(z)$,
$Q_{\lambda,k,0}(z)=Q_{\lambda,k}(z)$, then a variable
substitution $z-r\mapsto z$ shows that we can choose
$P_{\lambda,k,r}(z)=P_{\lambda,k}(z-r)$,
$Q_{\lambda,k,r}(z)=Q_{\lambda,k}(z-r)$. Another point that can be
noticed in equation (\ref{EqPr}) by taking the complex conjugate
is that we can choose $P_{\lambda,-k}(z)=\bar P_{\lambda,k}(z)$,
$Q_{\lambda,-k}(z)=\bar Q_{\lambda,k}(z)$. Finally we construct
the eigenfunctions
\begin{eqnarray}
\zeta_{k,\lambda,r,s}(\vec x)=(\det F(-r))e^{i\{F^\bot(r)\breve
k_0(-k),\breve x\}}P_{\lambda,k,s}(z-r),\label{Eigenfunc}
\end{eqnarray}
where to $s=1$ ($-1$) corresponds $P_{\lambda,k,s}=P_{\lambda,k}$
($Q_{\lambda,k}$). Note that each $\zeta_{k,\lambda,r,s}$ enters
with its conjugate,
$\bar\zeta_{k,\lambda,r,s}=\zeta_{-k,\lambda,r,s}$. As we will see
in the next section, $P_{\lambda,k,s}$ are orthogonal with respect
to the weight $\det F(-z)$, which shows that
$\zeta_{k,\lambda,r,s}$ just defined are orthogonal with respect
to the same weight. Again, instead of using the Plancherel measure
(\ref{dnu_k}) we can use $d\nu(k)=\dot\nu(k)dk$ and renormalize
according to
$$
\zeta_{k,\lambda,r,s}(\vec x)=\frac{1}{\sqrt{\mathfrak{h}}}(\det
F(-r))e^{i\{F^\bot(r)\breve k_0(-k),\breve
x\}}P_{\lambda,k,s}(z-r).
$$
Note that by (\ref{X_iExplicit}) the number $\mathfrak{h}$ is just
$\sqrt{\det\check h_{ij}}$.

\section{Fourier transform on Bianchi II-VII groups \label{Section_Fourier}}

As a first step on the way of establishing the completeness of
$\{\zeta_{k,\lambda,r,s}\}$ we prove a simple proposition.
Consider the differential operator
$$
D_{\breve k}=\check h^{33}\frac{d^2}{dz^2}+i(\breve k^\top
F(z)[\check h^{\bullet3}+(\check
h^{3\bullet})^\top])\frac{d}{dz}-(\breve k^\top F(z)\check
h^{2\times2}F^\top(z)\breve k-i\check
h^{3\bullet}F^\top(z)M^\top\breve k),
$$
$$
\breve k\in\mathbb{R}^2\setminus V^0,
$$
which by definition satisfies
$$
D_{\breve k}f(z)=e^{-i\{\breve k,\breve
x\}}\mathbf{\Delta}_h\left[e^{i\{\breve k,\breve
x\}}f(z)\right]\mbox{, }f\in C^\infty_0(\mathbb{R}).
$$

\begin{proposition} The operator $D_{\breve k}$ with domain $C^\infty_0(\mathbb{R})$ is symmetric in $L^2(\mathbb{R},\det F(-z)dz)$, for any $\breve
k\in\mathbb{R}^2\setminus V^0$.
\end{proposition}

\begin{proof}
Let us first write Green's identity for the operator
$\mathbf{\Delta}_h$ on the infinite tube
$D^1\times\mathbb{R}\subset G$ where $D^1$ is the unit disk in the
$\breve x$-plane,
$$
\int_{D^1\times\mathbb{R}}d\vec x\left(e^{-i\{\breve k,\breve
x\}}\bar g(z)\mathbf{\Delta}_h\left[e^{i\{\breve k,\breve
x\}}f(z)\right]-\mathbf{\Delta}_h\left[e^{-i\{\breve k,\breve
x\}}\bar g(z)\right]e^{i\{\breve k,\breve x\}}f(z)\right)
$$
$$
=\int_{S^1\times\mathbb{R}}dzdl(\breve x)\left(e^{-i\{\breve
k,\breve x\}}\bar g(z)(\breve x,\frac{\partial}{\partial\breve
x})\left[e^{i\{\breve k,\breve x\}}f(z)\right]-(\breve
x,\frac{\partial}{\partial\breve x})\left[e^{-i\{\breve k,\breve
x\}}\bar g(z)\right]e^{i\{\breve k,\breve x\}}f(z)\right)
$$
$$
=\int_{S^1\times\mathbb{R}}dzdl(\breve x)2i\bar g(z)f(z)(\breve
x,\breve k)=0.
$$
Next we note that
$$
\int_{D^1\times\mathbb{R}}d\vec x\left(e^{-i\{\breve k,\breve
x\}}\bar g(z)\mathbf{\Delta}_h\left[e^{i\{\breve k,\breve
x\}}f(z)\right]-\mathbf{\Delta}_h\left[e^{-i\{\breve k,\breve
x\}}\bar g(z)\right]e^{i\{\breve k,\breve x\}}f(z)\right)
$$
$$
=\int_{D^1\times\mathbb{R}}dxdydz(\det F(-z))\left(\bar
g(z)D_{\breve k}f(z)-\bar D_{\breve k}[\bar g(z)]f(z)\right)
$$
$$
=\pi\int_\mathbb{R}dz(\det F(-z))\left(\bar g(z)D_{\breve
k}f(z)-\bar D_{\breve k}[\bar g(z)]f(z)\right)=0,
$$
which holds on the dense subset of all $f,g$ in
$C^\infty_0(\mathbb{R})$ inside $L^2(\mathbb{R},\det F(-z)dz)$,
and symmetry is thus proven.$\Box$
\end{proof}

Now from the definition it is clear that $D_{\breve k}$ is a
negative definite operator (because $\mathbf{\Delta}_h$ is such),
and is hence upper semibounded, and has a self-adjoint extension
in $L^2(\mathbb{R},\det F(-z)dz)$ by Friedrichs extension theorem
\cite{Reed1975}. In particular, for $\breve k=\breve k_0(-k)$,
$k\in\mathfrak{K}$, the generalized eigenfunctions
$\{P_{\lambda,k,s}\}_{\lambda\in Sp(\mathbf{\Delta}_h),s=\pm1}$
are complete and give rise to a Fourier transform
$\mathfrak{F}_{\breve k_0(-k)}$ on $L^2(\mathbb{R},\det F(-z)dz)$
by means of an abstract eigenfunction expansion.
$\mathfrak{F}_{\breve k_0(-k)}$ is given by
$$
(\mathfrak{F}_{\breve k_0(-k)}f)(\lambda,s)=\int_\mathbb{R}dz(\det
F(-z))\bar P_{\lambda,k,s}(z)f(z).
$$
Define now the linear isomorphism $\mathfrak{V}:L^2(\mathbb{R})\to
L^2(\mathbb{R},\det F(-z)dz)$ by
$$
f(z)=[\mathfrak{V}\phi](z)=\phi(-z)(\det F(z))^{\frac{1}{2}}.
$$
This induces a Fourier transform
$\mathfrak{F}_k=\mathfrak{F}_{\breve k_0(-k)}\mathfrak{V}$ which
acts as
$$
(\mathfrak{F}_k\phi)(\lambda,s)=\int_\mathbb{R}dz(\det
F(z))^{\frac{1}{2}}P_{\lambda,k,s}(-z)\phi(z)\doteq\tilde\phi(\lambda,k,s).
$$
The inversion formula is given by
$$
\phi(z)=(\det
F(z))^{\frac{1}{2}}\sum_{s=\pm1}\int_{Sp(\mathbf{\Delta}_h)}d\lambda\tilde\phi(\lambda,k,s)\bar
P_{\lambda,k,s}(-z).
$$
Now we are in the position to show how $\zeta_{k,\lambda,r,s}$ are
related to the irreducible representations $T_{\breve k_0(k)}$.
Consider the following transformation on $f\in C^\infty_0(G)$,
$$
\tilde f(k,\lambda,r,s)=\int_Gdg\bar\zeta_{k,\lambda,r,s}(g)f(g),
$$
with eigenfunctions $\zeta_{k,\lambda,r,s}$ defined in {\bf
Section \ref{SecSpecScal}}. We will see that $\tilde
f(k,\lambda,r,s)$ are in some sense proportional to the matrix
columns of the operators $\hat f(T_{\breve k_0(k)})$. First we see
that
$$
\tilde f(k,\lambda,r,s)=\mathfrak{h}(\det
F(-r))\int_{\mathbb{R}^3}dxdydz(\det F(-z))
$$
$$
\times f(x,y,z)e^{i(F^\bot(r)\breve k_0(k),\breve x)}\bar
P_{\lambda,k,s}(-(r-z))
$$
$$
=\mathfrak{h}(\det F(-r))\int_{\mathbb{R}^3}dxdydz(\det
F(-z))(\det F(r-z))^{-\frac{1}{2}}f(x,y,z)e^{i(F^\bot(r)\breve
k_0(k),\breve x)}
$$
$$
\times(\det F(r-z))^{\frac{1}{2}}\bar P_{\lambda,k,s}(-(r-z)).
$$
Next we recognize that this is related to the extension of the
operator $\hat f(T_{\breve k_0(k)})$ from $L^2(\mathbb{R})$ to
$C^\infty(\mathbb{R})$,
$$
\tilde f(k,\lambda,r,s)=(\det F(-r))\hat f(T_{\breve
k_0(k)})\left[(\det F(r))^{\frac{1}{2}}\bar
P_{\lambda,k,s}(-r)\right].
$$
Integrating we obtain
\begin{eqnarray}
\sum_{s=\pm1}\int_{Sp(\mathbf{\Delta}_h)}d\lambda\tilde
f(k,\lambda,r,s)\tilde\phi(\lambda,k,s)=(\det F(-r))\hat
f(T_{\breve k_0(k)})\phi[r].\label{IrrepTrans}
\end{eqnarray}
Recall now the Fourier inversion formula as given in
\cite{Currey2005} (notation there is different, and we have
adapted them to ours adopted from \cite{Folland199502}),
\begin{eqnarray}
f(1)=\int_{\mathfrak{K}}d\nu(k)Tr\left[D_\pi\hat f(T_{\breve
k_0(k)})\right].\label{FourierInverse}
\end{eqnarray}
Formally a matrix element of $D_\pi\hat f(T_{\breve k_0(k)})$
would be an expression
$$
\left((\det F(z))^{\frac{1}{2}}\bar
P_{\lambda',k,s'}(-z),D_\pi\hat f(T_{\breve k_0(k)})(\det
F(z))^{\frac{1}{2}}\bar
P_{\lambda,k,s}(-z)\right)_{L^2(\mathbb{R})}
$$
$$
=\int_\mathbb{R}dz(\det F(z))P_{\lambda',k,s'}(-z)\tilde
f(k,\lambda,z,s),
$$
which does not make sense in precise terms. However, the trace of
such elements,
$$
\sum_{s=\pm1}\int_{Sp(\mathbf{\Delta}_h)}d\lambda\int_\mathbb{R}dz(\det
F(z))P_{\lambda,k,s}(-z)\tilde f(k,\lambda,z,s),
$$
can be given an exact sense if we change the order of integration,
$$
\int_\mathbb{R}dz\sum_{s=\pm1}\int_{Sp(\mathbf{\Delta}_h)}d\lambda
(\det F(z))P_{\lambda,k,s}(-z)\tilde f(k,\lambda,z,s).
$$
Indeed, let $\{p_n(z)\}$ be an orthonormal system in
$L^2(\mathbb{R})$. Consider the Fourier transforms $\tilde
p_n(\lambda,k,s)$, and consider the following bi-distribution in
the Fourier space,
$\sum_{n=1}^\infty\overline{\widetilde{p_n}}(\lambda,k,s)\widetilde{p_n}(\lambda',k,s')$.
Let $\tilde f,\tilde g$ be the Fourier transforms of arbitrary
$f,g\in L^2(\mathbb{R})$. We have
$$
\sum_{n=1}^\infty\sum_{s=\pm1}\sum_{s'=\pm1}\int_{Sp(\mathbf{\Delta}_h)}d\lambda\int_{Sp(\mathbf{\Delta}_h)}d\lambda'\overline{\widetilde{p_n}}(\lambda,k,s)\widetilde{p_n}(\lambda',k,s')\tilde
f(\lambda,k,s)\bar{\tilde{g}}(\lambda',k,s')
$$
$$
=\sum_{n=1}^\infty\left(\sum_{s=\pm1}\int_{Sp(\mathbf{\Delta}_h)}d\lambda\overline{\widetilde{p_n}}(\lambda,k,s)\tilde
f(\lambda,k,s)\right)\left(\sum_{s'=\pm1}\int_{Sp(\mathbf{\Delta}_h)}d\lambda'\widetilde{p_n}(\lambda',k,s')\bar{\tilde{g}}(\lambda',k,s')\right)
$$
$$
=\sum_{n=1}^\infty(p_n,f)_{L^2(\mathbb{R})}(g,p_n)_{L^2(\mathbb{R})}=(g,f)_{L^2(\mathbb{R})}=\sum_{s=\pm1}\int_{Sp(\mathbf{\Delta}_h)}d\lambda\tilde
f(\lambda,k,s)\bar{\tilde{g}}(\lambda,k,s),
$$
thus
$\sum_{n=1}^\infty\overline{\widetilde{p_n}}(\lambda,k,s)\widetilde{p_n}(\lambda',k,s')=\delta(\lambda-\lambda')\delta^s_{s'}$.
Now
$$
\int_\mathbb{R}dz\sum_{s=\pm1}\int_{Sp(\mathbf{\Delta}_h)}d\lambda
(\det F(z))P_{\lambda,k,s}(-z)\tilde f(k,\lambda,z,s)
$$
$$
=\int_\mathbb{R}dz(\det
F(z))\sum_{n=1}^\infty\sum_{s,s'}\underset{Sp(\mathbf{\Delta}_h)^2}{\int\!\int}
d\lambda
d\lambda'\overline{\widetilde{p_n}}(\lambda,k,s)\widetilde{p_n}(\lambda',k,s')P_{\lambda,k,s}(-z)\tilde
f(k,\lambda',z,s')
$$
$$
=\int_\mathbb{R}dz(\det
F(z))\sum_{n=1}^\infty\overline{\left(\sum_{s=\pm1}\int\limits_{Sp(\mathbf{\Delta}_h)}d\lambda\widetilde{p_n}(\lambda,k,s)\bar
P_{\lambda,k,s}(-z)\right)}
$$
$$
\times\left(\sum_{s'=\pm1}\int\limits_{Sp(\mathbf{\Delta}_h)}d\lambda'\tilde
f(k,\lambda',z,s')\widetilde{p_n}(\lambda',k,s')\right)
$$
using (\ref{IrrepTrans}),
$$
=\int_\mathbb{R}dz(\det F(z))\sum_{n=1}^\infty(\det
F(z))^{-\frac{1}{2}}\overline{p_n}(z)(\det F(-z))\hat f(T_{\breve
k_0(k)})p_n(z)
$$
as both the sum and the integral converge in $L^2$,
$$
=\sum_{n=1}^\infty\int_\mathbb{R}dz\overline{p_n}(z)(\det
F(z))^{-\frac{1}{2}}\hat f(T_{\breve
k_0(k)})p_n(z)=\sum_{n=1}^\infty\int_\mathbb{R}dz\overline{p_n}(z)D_\pi\hat
f(T_{\breve k_0(k)})p_n(z)
$$
$$
=\sum_{n=1}^\infty(p_n,D_\pi\hat f(T_{\breve
k_0(k)})p_n)_{L^2(\mathbb{R})}=Tr\left[D_\pi\hat f(T_{\breve
k_0(k)})\right].
$$
Hence from (\ref{FourierInverse}) we have
$$
f(1)=\int_{\mathfrak{K}}d\nu(k)\int_\mathbb{R}dz\sum_{s=\pm1}\int_{Sp(\mathbf{\Delta}_h)}d\lambda
(\det F(z))P_{\lambda,k,s}(-z)\tilde f(k,\lambda,z,s).
$$
To find an inversion formula at an arbitrary point $g\in G$ we
apply this to the left translated function
$[L_{g^{-1}}f](x)=f(gx)$,
$$
f(g)=[L_{g^{-1}}f](1)=\int_{\mathfrak{K}}d\nu(k)\int_\mathbb{R}dz\sum_{s=\pm1}\int_{Sp(\mathbf{\Delta}_h)}d\lambda
(\det
F(z))P_{\lambda,k,s}(-z)\widetilde{[L_{g^{-1}}f]}(\lambda,k,z,s).
$$
But from the definition
$$
\widetilde{[L_{g^{-1}}f]}(\lambda,k,r,s)=\int_Gdh\bar\zeta_{k,\lambda,r,s}(h)[L_{g^{-1}}f](h)=\int_Gdh'\bar\zeta_{k,\lambda,r,s}(g^{-1}h')f(h').
$$
From the definition of $\zeta_{k,\lambda,r,s}$ we find
$$
\bar\zeta_{k,\lambda,r,s}(g^{-1}h')=e^{-i(F^\bot(r+g_z)\breve
k_0(k),\breve g)}(\det F(g_z))\bar\zeta_{\lambda,k,r+g_z,s}(h'),
$$
thus
$$
\int_Gdh'\bar\zeta_{k,\lambda,r,s}(g^{-1}h')f(h')=e^{-i(F^\bot(r+g_z)\breve
k_0(k),\breve g)}(\det F(g_z))\tilde f(k,\lambda,r+g_z,s).
$$
Therefore
$$
f(g)=\int_{\mathfrak{K}}d\nu(k)\int_\mathbb{R}dz\sum_{s=\pm1}\int_{Sp(\mathbf{\Delta}_h)}d\lambda
(\det F(z))P_{\lambda,k,s}(-z)
$$
$$
\times e^{-i(F^\bot(z+g_z)\breve k_0(k),\breve g)}(\det
F(g_z))\tilde f(k,\lambda,z+g_z,s)
$$
by substitution $r=z+g_z$
$$
=\int_{\mathfrak{K}}d\nu(k)\int_\mathbb{R}dr\sum_{s=\pm1}\int_{Sp(\mathbf{\Delta}_h)}d\lambda
(\det F(r))\tilde f(k,\lambda,r,s)
$$
$$
\times e^{-i(F^\bot(r)\breve k_0(k),\breve
g)}P_{\lambda,k,s}(g_z-r)
$$
\begin{eqnarray}
=\int_{\mathfrak{K}}d\nu(k)\int_\mathbb{R}dr\sum_{s=\pm1}\int_{Sp(\mathbf{\Delta}_h)}d\lambda
(\det F(r))\tilde
f(k,\lambda,r,s)\zeta_{k,\lambda,r,s}(g),\label{Fourier_Inv}
\end{eqnarray}
which is our final inversion formula.

It remains to note that by denoting $\alpha=(k,\lambda,r,s)$ we
have satisfied all conditions for the eigenfunction expansion
$\bar\zeta_\alpha(f)$ to give a conventional Fourier transform in
sense of \cite{ZhAPub12012}.

\section{Automorphism groups of Bianchi I-VII groups \label{Section_Aut}}

In this section we consider the automorphism groups $\Aut(G)$ of
Bianchi I-VII groups. After performing the calculations we
discovered that these automorphisms have been obtained earlier in
\cite{Harvey1979}. However we give here also the dual actions of
these automorphisms on $\hat G$ which is new. This may become
important when analyzing the transformation in the Fourier space
induced by automorphisms. We start by noting that Bianchi I-VII
groups are matrix groups, and their matrix realization can be
given by
$$
\begin{pmatrix}
x \\ y \\ z
\end{pmatrix}
\mapsto G(x,y,z)=
\begin{pmatrix}
F(z) &

\begin{matrix}
x \\ y
\end{matrix}\\

\begin{matrix}
0 & 0
\end{matrix} & 1
\end{pmatrix}.
$$
It can be easily seen that in this realization the group
multiplication indeed corresponds to the matrix multiplication.
The respective Lie algebra realization will be
$$
\begin{pmatrix}
x \\ y \\ z
\end{pmatrix}
\mapsto\mathfrak{g}(x,y,z)=
\begin{pmatrix}
zM &

\begin{matrix}
x \\ y
\end{matrix}\\

\begin{matrix}
0 & 0
\end{matrix} & 0
\end{pmatrix},
$$
which again can be checked to intertwine the matrix commutation
with the Lie bracket. Moreover, we could have obtained immediately
the exponential map by setting
$\exp(x,y,z)=\exp(\mathfrak{g}(x,y,z))$ instead of referring to
the Zassenhaus formula, but the latter is a more Lie theoretical
approach. Now that all Bianchi groups are connected and simply
connected by Theorem 1 of III.6.1 in \cite{Bourbaki1998} it
follows that $\Aut(G)=\Aut(\mathfrak{g})$ in the sense of a
topological group isomorphism (see also \cite{Hochschild1965}). An
algebra homomorphism of matrix algebras is necessarily linear in
the matrix elements. It follows that any
$\check\alpha\in\Aut(\mathfrak{g})$ depends linearly on $x$, $y$,
$z$, and is therefore given by an affine transformation in
$\mathbb{R}^3$, which is actually a linear transformation because
it preserves 0. Therefore we first determine $\Aut(\mathfrak{g})$.
Let the linear map $\check\alpha:\mathbb{R}^3\to\mathbb{R}^3$ be
given by
$$
\begin{pmatrix}
x \\ y \\ z
\end{pmatrix}
=
\begin{pmatrix}
\check\alpha_{2\times 2} & \check\alpha_{\bullet3} \\
\check\alpha_{3\bullet} & \check\alpha_{33}
\end{pmatrix}
\begin{pmatrix}
q \\ r \\ s
\end{pmatrix},
$$
where we use notation similar to (\ref{MatrixNot}). Then
$\check\alpha\in\Aut(\mathfrak{g})$ if and only if
$\check\alpha[\vec x,\vec y]=[\check\alpha\vec x,\check\alpha\vec
y]$, where $[,]$ is the Lie bracket. Expanding this condition we
get the system of requirements
\begin{eqnarray}
\check\alpha_{2\times 2}M-\check\alpha_{33}M\check\alpha_{2\times
2}+M\check\alpha_{\bullet3}\check\alpha_{3\bullet}=0,\label{alphaMcomm}
\end{eqnarray}
$$
\check\alpha_{2\times 2}M\sigma\check\alpha_{3\bullet}^\top=0,
$$
$$
\check\alpha_{3\bullet}M=0,
$$
where $\sigma$ is the unit antisymmetric matrix. The patterns of
admissible matrices $\check\alpha$ satisfying this system have to
be computed for each group independently. For Bianchi I we have
$M=0$ and all three conditions are satisfied trivially. For
Bianchi IV-VII the matrix $M$ is invertible hence the third
requirement means $\check\alpha_{3\bullet}=0$, so that the second
becomes trivial, and the first reduces to $\check\alpha_{2\times
2}M-\check\alpha_{33}M\check\alpha_{2\times 2}=0$. The cases of
groups Bianchi II and III are a bit more involved, but the
calculations are straightforward. We present the results in the
Table \ref{AutomTable}. Note that whenever
$\check\alpha_{3\bullet}=0$ the invertibility of $\check\alpha$
requires $\check\alpha_{33}\neq0$. As it can be seen from the
table some algebras allow for reflective automorphisms and their
automorphism groups consist of two components (this is what the
union symbol $\bigcup$ in Table \ref{AutomTable} refers to).
Matrices of these pattern forms exhaust the groups
$\Aut(\mathfrak{g})$. One can compare this pattern of
automorphisms to those available in the literature, for instance,
of the Heisenberg algebra in \cite{Folland1989}.

\begin{table}\label{AutomTable}
\begin{tabular}{|c|c|c|c|c|c|}\hline
I & II & III & IV & V & VI, $q\neq1$\\\hline

$\begin{pmatrix} a & b & c\\
d & e & f\\
g & h & j
\end{pmatrix}$ &

$\begin{pmatrix} a & 0 & c\\
d & a\cdot j-c\cdot g & f\\
g & 0 & j
\end{pmatrix}$ &

$\begin{pmatrix} a & 0 & c\\
0 & e & f\\
0 & 0 & 1
\end{pmatrix}$ &

$\begin{pmatrix} a & 0 & c\\
d & a & f\\
0 & 0 & 1
\end{pmatrix}$ &

$\begin{pmatrix} a & b & c\\
d & e & f\\
0 & 0 & 1
\end{pmatrix}$ &

$\begin{pmatrix} a & 0 & c\\
0 & e & f\\
0 & 0 & 1
\end{pmatrix}$\\\hline

\end{tabular}

\begin{tabular}{|c|c|c|}\hline
VI, $q=1$ & VII, $p\neq0$ & VII, $p=0$\\\hline

$\begin{pmatrix} a & 0 & c\\
0 & e & f\\
0 & 0 & 1
\end{pmatrix}
\bigcup
\begin{pmatrix} 0 & b & c\\
d & 0 & f\\
0 & 0 & -1
\end{pmatrix}$ &

$\begin{pmatrix} a & b & c\\
-b & a & f\\
0 & 0 & 1
\end{pmatrix}$ &

$\begin{pmatrix} a & b & c\\
-b & a & f\\
0 & 0 & 1
\end{pmatrix}
\bigcup
\begin{pmatrix} a & b & c\\
b & -a & f\\
0 & 0 & -1
\end{pmatrix}$\\\hline

\end{tabular}
\caption{Patterns of permissible matrices $\check\alpha$ for
Bianchi I-VII algebras}
\end{table}

Now the corresponding group homomorphisms $\check A\in\Aut(G)$ can
be found by composing $\check\alpha\in\Aut(\mathfrak{g})$ with the
exponential map, $\check
A\exp((x,y,z))=\exp(\check\alpha(x,y,z))$. Recall that the
exponential map is given by
$$
\exp((x,y,z))=([1+F(z)D(z)](x,y),z),
$$
(where $D(z)$ is defined in (\ref{D_ZDef}))and because this map is
bijective we know that the matrix $[1+F(z)D(z)]$ is invertible for
all $z$. The logarithmic map can be written as
$$
\log((x,y,z))=([1+F(z)D(z)]^{-1}(x,y),z),
$$
and the action of the group homomorphism $\check A$ related to the
algebra homomorphism $\check\alpha$ becomes
$$
\check A
\begin{pmatrix}\breve x\\z\end{pmatrix}
=
\begin{pmatrix}[1+F(z')D(z')]\left(\check\alpha_{2\times 2}[1+F(z)D(z)]^{-1}\breve
x+\check\alpha_{\bullet3}z\right)\\z'\end{pmatrix}
$$
$$
z'=\check\alpha_{3\bullet}[1+F(z)D(z)]^{-1}\breve
x+\check\alpha_{33}z,
$$
for Bianchi II-VII groups and
$$
\check A
\begin{pmatrix}\breve x\\z\end{pmatrix}
= \check\alpha
\begin{pmatrix}\breve x\\z\end{pmatrix}
$$
for the Bianchi I group. From $\check\alpha_{3\bullet}M=0$ it
follows that
$\check\alpha_{3\bullet}[1+F(z)D(z)]^{-1}=\check\alpha_{3\bullet}$.
Thus the formula for Bianchi II-VII simplifies to
$$
\check A
\begin{pmatrix}\breve x\\z\end{pmatrix}
=
\begin{pmatrix}[1+F(\check\alpha_{3\bullet}\breve
x+\check\alpha_{33}z)D(\check\alpha_{3\bullet}\breve
x+\check\alpha_{33}z)]\left(\check\alpha_{2\times
2}[1+F(z)D(z)]^{-1}\breve
x+\check\alpha_{\bullet3}z\right)\\\check\alpha_{3\bullet}\breve
x+\check\alpha_{33}z\end{pmatrix}.
$$
One more step can be done in this generality. From formula
(\ref{alphaMcomm}) and $\check\alpha_{3\bullet}M=0$ it follows
that
$$
\check\alpha_{2\times
2}M^m=(\check\alpha_{33}M)^m\check\alpha_{2\times 2}
$$
for $m\ge2$ and therefore for any sequence of complex numbers
$\{f_m\}_{m=0}^\infty$
$$
\check\alpha_{2\times 2}\sum_{m=0}^\infty f_mM^m=\sum_{m=0}^\infty
f_m(\check\alpha_{33}M)^m\check\alpha_{2\times
2}+f_1M\check\alpha_{\bullet3}\check\alpha_{3\bullet}
$$
whenever the left hand side exists. This can be used to establish
that
$$
[1+F(\check\alpha_{3\bullet}\breve
x+\check\alpha_{33}z)D(\check\alpha_{3\bullet}\breve
x+\check\alpha_{33}z)]\check\alpha_{2\times
2}=\check\alpha_{2\times
2}[1+F(\frac{\check\alpha_{3\bullet}\breve
x}{\check\alpha_{33}}+z)D(\frac{\check\alpha_{3\bullet}\breve
x}{\check\alpha_{33}}+z)].
$$
This far on the explicit form of the group automorphisms.

Now let us look at the dual spaces $\hat G$. If $\check
A\in\Aut(G)$ and $\pi\in\hat G$ then $\pi\circ\check A=\pi'$ for
some $\pi'\in\hat G$. Thus $\check A$ induces a pullback map
$\check A^*:\hat G\to\hat G$. Because $\dim\pi=\dim\pi'$ it
follows that $\check A^*$ maps generic representations into
generic representations and singletons into singletons. The
representations $\pi\in\hat G$ are in a bijective correspondence
with the derived representations $d\pi$ which are irreducible
representations of the Lie algebra $\mathfrak{g}$. In a similar
fashion, any $\check\alpha\in\Aut(\mathfrak{g})$ induces a
pullback map $\alpha^*:d\hat G\to d\hat G$ between derived
representations. This pullback map is easier to study than that
for the group representations. Consider first the Bianchi I group.
The irreducibles are given by
$$
T_{\vec k}(\vec g)=e^{i\{\vec k,\vec g\}},
$$
and the derived representations are
$$
dT_{\vec k}(\vec x)=i\{\vec k,\vec x\}.
$$
An automorphism $\vec x=\check\alpha\vec q$ induces the pullback
map $\check\alpha^*(\vec k)=\check\alpha^\top\vec k$. Consider now
the singletons of a Bianchi II-VII group. They are given for $\vec
k\in V^0\oplus\mathbb{R}$ by
$$
T_{\vec k}(\vec g)=e^{i\{\vec k,\vec g\}}=e^{i\{\breve k,\breve
g\}}e^{ik_3g_z},
$$
and the derived singletons are
$$
dT_{\vec k}(\vec x)=i\{\vec k,\vec x\},
$$
and again, an automorphism $\vec x=\check\alpha\vec q$ induces the
pullback map $\check\alpha^*(\vec k)=\check\alpha^\top\vec k$.
This in particular means that $\breve k'=\check\alpha_{2\times
2}^\top\breve k+k_3\check\alpha_{3\bullet}^\top$, and if $\breve
k\in V^0$ then
$$
M^\top\breve k'=M^\top\check\alpha_{2\times 2}^\top\breve k+k_3
M^\top\check\alpha_{3\bullet}^\top=0,
$$
where (\ref{alphaMcomm}) and $\check\alpha_{3\bullet}M=0$ were
used. We explicitly observe that the automorphisms map singletons
into singletons, as expected. Finally we turn to the generic
representations. Let $T_{\breve k}$ be a generic representation of
$G$. Then it acts on $L^2(\mathcal{R})$ by
$$
T_{\breve k}(\vec g)f[w]=e^{i\{\breve k,F(-w)\breve
g\}}f[w-g_z]\mbox{, }\vec g=(\breve g,g_z)\in G.
$$
Its derived representation will be
$$
dT_{\breve k}(\vec x)f[w]=i\{\breve k,F(-w)\breve
x\}f[w]-z\partial_wf[w].
$$
Under the automorphism $\vec x=\check\alpha\vec q$ it will turn
into
$$
dT_{\breve k}(\vec q)f[w]=i\{\breve k,F(-w)[\check\alpha_{2\times
2}\breve
q+\check\alpha_{\bullet3}s]\}f[w]-[\check\alpha_{3\bullet}\breve
q+\check\alpha_{33}s]\partial_wf[w].
$$
For simplicity we will consider only the automorphisms with
$\check\alpha_{3\bullet}=0$. Thus we omit only some automorphisms
of the Heisenberg group, but this group is a central subject in
harmonic analysis, and the missing results can be found in the
literature. Define the isometric isomorphism
$\mathfrak{T}:L^2(\mathbb{R})\to L^2(\mathbb{R})$ by
$$
\mathfrak{T}(f)[w]=\frac{1}{\sqrt{\check\alpha_{33}}}e^{i\{\breve
k,\int_0^{-w}F(\check\alpha_{33}\xi)d\xi\check\alpha_{\bullet3}\}}f(\check\alpha_{33}w).
$$
Consider the representation $dT_{\breve k'}$ with $\breve
k'=\check\alpha_{2\times 2}^\top\breve k$. Note that because
$\check\alpha_{3\bullet}=0$ we have that $\check\alpha_{2\times
2}^\top$ is invertible, and from (\ref{alphaMcomm}) we assure that
it maps $\breve k\notin V^0$ to $\breve k'\notin V^0$. Thus
$dT_{\breve k'}$ is generic. Its action on the image
$\mathfrak{T}(f)[w]$ is given by
$$
dT_{\breve k'}(\vec q)\mathfrak{T}(f)[w]=i\{\breve k',F(-w)\breve
q\}\mathfrak{T}(f)[w]-s\partial_w\mathfrak{T}(f)[w]
$$
$$
=\mathfrak{T}(i\{\breve
k',F(-\frac{\bullet}{\check\alpha_{33}})\breve
q\}f)[w]+\mathfrak{T}(is\{\breve
k,F(-\bullet)\check\alpha_{\bullet3}\}f)[w]-\mathfrak{T}(s\check\alpha_{33}\partial
f).
$$
Recall that we have seen that from (\ref{alphaMcomm}) and
$\check\alpha_{3\bullet}=0$ it follows $\check\alpha_{2\times
2}F(z)=F(\check\alpha_{33}z)\check\alpha_{2\times 2}$, hence
$$
i\{\breve k',F(-\frac{\bullet}{\check\alpha_{33}})\breve
q\}=i\{\check\alpha_{2\times 2}^\top\breve
k,F(-\frac{\bullet}{\check\alpha_{33}})\breve q\}=i\{\breve
k,F(-\bullet)\check\alpha_{2\times 2}\breve q\}.
$$
We finally see that
$$
dT_{\breve k'}(\vec
q)\mathfrak{T}(f)[w]=\mathfrak{T}\left([i\{\breve
k,F(-\bullet)\check\alpha_{2\times 2}\breve q\}+is\{\breve
k,F(-\bullet)\check\alpha_{\bullet3}\}]f-s\check\alpha_{33}\partial
f\right)=\mathfrak{T}\left(dT_{\breve k}(\vec q)f\right),
$$
which means that $\mathfrak{T}$ intertwines the irreducible
representations $dT_{\breve k}\circ\check\alpha$ and
$dT_{\check\alpha_{2\times 2}^\top\breve k}$. Thus these two
representations are unitarily equivalent, $\check\alpha^*(\breve
k)=\check\alpha_{2\times 2}^\top\breve k$. If the cross sections
are chosen explicitly (for instance, as we did) then it is a
straightforward calculation to find the action of $\check\alpha^*$
on $\tilde K$ and $\mathfrak{K}$. We omit these calculations here
because, first, they depend on the preferred choice of the cross
sections, and second, they involve transcendental functions (e.g.,
the solution of the equation $e^y+ay=x$) and are not transparent
visually, and do not provide a better insight into the matter.

\section{Separation of time variable in homogeneous universes \label{Section_SepVarHomUni}}

We want to see to which extent the technique of mode decomposition
developed in \cite{ZhAPub12012} is applicable to hyperbolic fields
on Bianchi type and FRW cosmological models. For this aim we have
to check whether the conditions of {\bf Proposition 2.3} are
satisfied. Recall that the metric $g$ of a homogeneous spacetime
$M=\mathcal{I}\times\Sigma$, where $\mathcal{I}$ is an open
interval and $\Sigma$ is a Bianchi type homogeneous space, is
given by
$$
ds^2_g=dt^2-\sum_{\alpha,\beta=1}^3\check
h_{\alpha\beta}(t)d\omega^\alpha(\vec x)d\omega^\beta(\vec x),
$$
where $\check h_{ij}(t)$ ($t\in\mathcal{I}$) is a smooth positive
definite symmetric matrix function, and $d\omega^i$ are the left
invariant 1-forms on $\Sigma$. Condition (i) of {\bf Proposition
2.3} is automatically satisfied because $g_{00}=1$. For condition
(ii) note that
$$
\sum_{i,j=1}^3g^{ij}(x)\frac{\partial g_{ij}}{\partial
t}(x)=\sum_{i,j=1}^3\sum_{\alpha,\beta=1}^3\sum_{\gamma,\delta=1}^3\check
h^{\alpha\beta}(t)\dot{\check{h}}_{\gamma\delta}(t)X_\alpha^i(\vec
x)X_\beta^j(\vec x)(d\omega^\gamma)_i(\vec
x)(d\omega^\delta)_j(\vec x)
$$
$$
=\sum_{\alpha,\beta=1}^3\sum_{\gamma,\delta=1}^3\check
h^{\alpha\beta}(t)\dot{\check{h}}_{\gamma\delta}(t)\langle
X_\alpha,d\omega^\gamma\rangle_h\langle
X_\beta,d\omega^\delta\rangle_h
$$
$$
=\sum_{\alpha,\beta=1}^3\sum_{\gamma,\delta=1}^3\check
h^{\alpha\beta}(t)\dot{\check{h}}_{\gamma\delta}(t)\delta^\gamma_\alpha\delta^\delta_\beta=Tr[\check
h^{-1}(t)\dot{\check{h}}(t)].
$$
This shows that the condition (ii) is also satisfied. We see that
the homogeneous spacetimes are an ideal playground for mode
decomposition. Note that FRW spacetimes correspond to the choice
$\check h_{\alpha\beta}(t)=a^2(t)\delta_{\alpha\beta}$.

If conditions (iii) and (iv) are also satisfied depends on the
chosen connection $\nabla$. For the scalar field (iii) is
automatically satisfied with $\Gamma=0$. In \cite{ZhAPub12012} we
defined the field operator $D=\Box+m^2(x)$, and the instantaneous
field operator $D_{\Sigma_t}=-\Delta+m^2(x)$. The eigenfunction
equation of the latter operator was written as
$D_{\Sigma_t}\zeta_\alpha=\lambda_t(\alpha)\zeta_\alpha$.
Condition (iv) of {\bf Proposition 2.3} in \cite{ZhAPub12012} can
be obviously satisfied if $\check h_{\alpha\beta}(t)=a^2(t)\check
h^0_{\alpha\beta}$ as in this case the time evolution amounts only
to a rescaling of $\lambda(\alpha)$ in
$D_{\Sigma_t}\zeta_\alpha=\lambda_t(\alpha)\zeta_\alpha$. Note
that because $D_{\Sigma_t}$ is $G$-invariant, the term $m^2$ is a
function of $t$ only. This is the situation where the dynamics of
the universe consists of merely an isotropic rescaling. Thus, for
instance, in case of FRW spacetimes condition (iv) is satisfied
automatically.

But the condition (iv) can be also satisfied non-trivially with an
anistropic rescaling and even some shears and rotations. This is
clearly possible for Bianchi I group, because the eigenfunctions
do not depend on the matrix $\check h$. For Bianchi II-VII groups
one has to look at the equation (\ref{EqPr}) to see to which
extent the solution $P(z)$ depends on the matrix $\check h$.
Suppose $\check h$ and $\check j$ are two matrices for which there
exist two common linearly independent solutions $P(z)$ and $Q(z)$.
Because we have already seen that an isotropic rescaling is always
possible, without loss of generality we assume $\check
h^{33}=\check j^{33}$ (we again use the notations
\ref{MatrixNot}). Fix $\breve k\in\mathbb{R}^2\setminus V^0$ and
$0<\lambda\in\mathbb{R}$. Now the condition that the two equations
have the same solution spaces can be cast into the following pair
of equations,
$$
\check h^{3\bullet}F^\top(z)\breve k=\check
j^{3\bullet}F^\top(z)\breve q,
$$
$$
\lambda+\breve k^\top F(z)\check h^{2\times2}F^\top(z)\breve
k=\lambda'+\breve q^\top F(z)\check j^{2\times2}F^\top(z)\breve q
$$
for some $\breve q\in\mathbb{R}^2\setminus V^0$,
$0<\lambda\in\mathbb{R}$ and for all $z\in\mathbb{R}$. That
non-trivial possibilities exist is clear visually, but we will not
go into details here. Once these conditions are satisfied at all
$t\in\mathcal{I}$ for the 1-parameter family of matrices $\check
h(t)$ describing the evolution of the spatial metric then
condition (iv) is satisfied, and we have an explicit formula for
the time dependent eigenvalue $\lambda_t(\alpha)$.

As electromagnetism is of primary importance for us, let us
finally show that the assumption $\check
h_{\alpha\beta}(t)=a^2(t)\check h^0_{\alpha\beta}$ is sufficient
to satisfy condition (iii) for the 1-form field. Indeed, the
1-form field is given by the Levi-Civita connection, for which the
connection forms are $(\Gamma_i)^a_b=-{\bf\Gamma}^a_{ib}$. Let us
compute the symbol ${\bf\Gamma}^a_{0b}$. It is easy to see that
${\bf\Gamma}^0_{0b}={\bf\Gamma}^a_{00}=0$. For $a,b>0$ we have
$$
{\bf\Gamma}^a_{0b}=\frac{1}{2}\sum_{m=1}^3g^{am}\frac{\partial
g_{mb}}{\partial t}=\sum_{\alpha,\beta,\gamma=1}^3\check
h^{\alpha\beta}(t)\dot{\check{h}}_{\beta\gamma}(t)X^a_\alpha(\vec
x)d\omega^\gamma_b(\vec x).
$$
If $\check h_{\alpha\beta}(t)=a^2(t)\check h^0_{\alpha\beta}$ then
$\partial_t\check{h}_{\alpha\beta}(t)=2H(t)\check
h_{\alpha\beta}(t)$ with $H(t)=\frac{\dot a(t)}{a(t)}$ being the
Hubble constant, and we get
$$
{\bf\Gamma}^a_{0b}=2H(t)\sum_{\alpha,\beta,\gamma=1}^3\check
h^{\alpha\beta}(t)\check h_{\beta\gamma}(t)X^a_\alpha(\vec
x)d\omega^\gamma_b(\vec x)=2H(t)\delta^a_b.
$$
Thus $\Gamma_0=-2H(t)0\oplus {\bf 1}_3$ is not only a function of
$t$, but also commutes with any matrix, hence (iii) is trivially
satisfied.

\section{The mode decomposition of the Klein-Gordon field \label{Section_KG}}

The investigation of classical and quantum fields in anisotropic
cosmological models has been carried out by different authors
since decades (see, e.g.,
\cite{Lukash1976},\cite{ZeldovichStarobinsky71},\cite{Fulling1974a},\cite{Fulling1974}).
However these works mainly concentrate on Bianchi I models where
the harmonic analysis and mode decomposition are obvious. With the
background developed above we can extend this to all Bianch I-VII
models.

Mode decomposition is the rigorous generalization of the well
known heuristic idea of decomposition of a linear field into
harmonic oscillators. In \cite{ZhAPub12012} we have described how
the mode decomposition can be performed explicitly for an
arbitrary vector valued field given the explicit spectral theory
of the model spatial sections $\Sigma$. As we have already
explicitly constructed the spectral theory of the line bundle over
Bianchi I-VII spacetimes, we can apply the mode decomposition to
the Klein-Gordon field and see what can be gained by this
technique.

Let $M$ be Bianchi I-VII type spacetime, i.e., a 4-dimensional
smooth globally hyperbolic Lorentzian manifold with a smooth
global time function chosen \cite{Bernal_Sanchez_2005}, and with
the isometry group $G$ which is one of the groups Bianchi I-VII,
so that $G$ acts simply transitively on the equal time
hypersurfaces $\Sigma_t$ (for all missing details see
\cite{ZhAPub12012}). Or to put it in simpler words, let
topologically $M=\mathbb{R}\times G$ where $G$ is a Bianchi I-VII
group, and let the metric be given by $ds^2=dt^2-h_{ij}(t,\vec
x)dx^idx^j$, $\vec x=(x^1,x^2,x^3)\in\Sigma_t$, so that for any
$t\in\mathbb{R}$ the Riemannian metric $h_{ij}(t,.)$ is left
invariant under the action of the underlying Lie group $G$. In
\cite{ZhAPub12012} we cosnidered the field operator $D=\Box+m^2$,
where $\Box$ is the d'Alembert operator related to the Levi-Civita
connection, and $m^2\in\mathbb{R}_+$ is a positive constant. The
Klein-Gordon field is described by the equation
$$
D\phi=(\Box+m^2)\phi=0.
$$
As we have seen in the previous section for each Bianchi type
there are certain restrictions on the dynamics of the spatial
metric $h(t)$ for the mode decomposition to be applicable. Recall
that $h(t)$ is described by the positive definite matrix $\check
h(t)$. For simplicity let us consider only an isotropic rescaling,
$$
\check h(t)=a^2(t)\check h(0).
$$
We further wrote in \cite{ZhAPub12012} $D=D_t+D_{\Sigma_t}$, where
$D_t$ is a differential operator in variable $t$, and
$D_{\Sigma_t}$ is the instantaneous field operator
$D_{\Sigma_t}=-\Delta_t+m^2$. Its eigenfunctions satisfying
$D_{\Sigma_t}\zeta_\alpha=\lambda_\alpha(t)\zeta_\alpha$ will be
the eigenfunctions of the Laplace operator,
$-\Delta_t\zeta_\alpha=(\lambda_\alpha(t)-m^2)\zeta_\alpha$.
Because the time dependent Fourier transform is normalized at time
$t=0$, we have $\lambda_\alpha(0)=-\lambda+m^2$ where
$\alpha=(k,\lambda,r,s)$. Using $\check h^{ij}(t)=a^{-2}(t)\check
h^{ij}(0)$ it follows that $\Delta_t=a^{-2}(t)\Delta_0$ and hence
\begin{eqnarray}
\lambda_\alpha(t)=\frac{-\lambda}{a^2(t)}+m^2.\label{lambda_alpha_t}
\end{eqnarray}
As we have demanded that the spacetime $M=\mathcal{I}\times\Sigma$
is globally hyperbolic, the Klein-Gordon field operator
$D=\Box+m^2$ has unique advanced(+) and retarded(-) fundamental
solutions $E^\pm:C_0^\infty(M)\mapsto C_0^\infty(M)$ with the
properties
\begin{enumerate}
\item $(\Box+m^2)E^\pm f=f=E^\pm(\Box+m^2) f$

\item $\supp\{E^\pm f\}\subset J^\pm(\supp f)$
\end{enumerate}
for all $f\in C_0^\infty(M)$. Here, $J^\pm(N)$ denotes the causal
future(+) and past(-) of a subset $N\subset M$. We refer to
\cite{BarGinouxPfaffle200703} for full discussion and further
references. Then $E=E^+-E^-$ is called the causal propagator of
the Klein-Gordon operator $\Box+m^2$ on $(M,g)$. Any $\phi=Ef$,
$f\in C_0^\infty(M)$ is a solution of the homogeneous Klein-Gordon
equation $(\Box+m^2)\phi=0$, and the restriction of $\phi$ to any
Cauchy surface is compactly supported. We define
$Sol_0(M)=EC_0^\infty(M)$. We also write $E(f,h)=\langle
f,Eh\rangle_{L^2(M)}$ where
$$
\langle f,h\rangle_{L^2(M)}=\int_Mf(x)h(x)dvol_g(x);
$$
$dvol_g$ is the volume form on $M$ induced by the metric $g$.
Moreover, we set $\mathcal{K}=C_0^\infty(M)/\ker E$. Then (the
real part of) $\mathcal{K}$ becomes a symplectic space with
symplectic form
$$
\sigma([f],[h])=E(f,h),\qquad [f]=f+\ker E,\qquad [h]=h+\ker E.
$$
The map $\mathcal{K}\mapsto Sol_0(M)$, $[f]\mapsto Ef$ is a
symplectomorphism upon endowing $Sol_0(M)$ with the symplectic
form
$$
\sigma(\phi,\psi)=\int_\mathcal{C}\left(\phi n^a\nabla_a\psi-\psi
n^a\nabla_a\phi\right)d\eta_\mathcal{C}
$$
for any Cauchy surface $\mathcal{C}$ in $M$ having future-pointing
unit normal field $n^a$ and metric-induced hypersurface measure
$d\eta_\mathcal{C}$. Again, we refer to
\cite{BarGinouxPfaffle200703} for a complete discussion and full
proofs.

Now by {\bf Proposition 2.3} in \cite{ZhAPub12012} any $\phi\in
Sol_0(M)$ can be written as
\begin{eqnarray}
\phi(t,\vec
x)=\int_{\tilde\Sigma}d\mu(\alpha)\left[a^\phi(\alpha)T_\alpha(t)\zeta_\alpha(\vec
x)+b^\phi(\alpha)\bar T_\alpha(t)\zeta_\alpha(\vec
x)\right],\label{ModeDecompKG}
\end{eqnarray}
where $\alpha=(k,\lambda,r,s)$ and $d\mu(\alpha)=d\nu(k)d\lambda
F(r)dr$. The modes $T_\alpha$ are to this point arbitrary
$\mu$-measurable solutions of the mode equation
\begin{eqnarray}
\ddot T_\alpha(t)+F(t)\dot
T_\alpha(t)+G_\alpha(t)T_\alpha(t)=0\label{ModeEqKG}
\end{eqnarray}
such that $T_\alpha$ and $\bar T_\alpha$ are linearly independent
solutions. As found in \cite{ZhAPub12012} for the scalar field on
an isotropically expanding universe,
\begin{eqnarray}
F(t)=P(t)=3H(t)=3\frac{\dot a(t)}{a(t)}\mbox{,
}G_\alpha(t)=\lambda_\alpha(t)\mbox{,
}I(t)=a^3(t).\label{ModeEqCoeff}
\end{eqnarray}
The propagator $E[f]$ is obviously a weak solution of the
Klein-Gordon equation. If we want to mode decompose it we have to
satisfy (2.19) of \cite{ZhAPub12012}. By {\bf Proposition 2.6} of
\cite{ZhAPub12012} we can do it if $M$ is an analytic manifold.
But the spatial metric $h$ is analytic in the spatial variable $x$
because it is expressed in left invariant fields of the Lie group
$G$. Thus we only need to choose $a(t)$ a real analytic function.
The spectrum of $D_{\Sigma_t}$ is strictly uniform with the
prescription $\omega(\alpha)=-\lambda$. Now if we restrict the
initial data $T_\alpha(0)=p(\lambda)$ and $\dot
T_\alpha(0)=q(\lambda)$ where $p,q\in\mathcal{A}(\mathbb{H}_0)$
and $p\bar q-\bar pq=i$ then {\bf Proposition 2.9} is applicable.
Suppose this is done, now from the {\bf Section 2.6} of
\cite{ZhAPub12012} we find that (for the line bundle obviously
$s(\alpha)=1$)
\begin{eqnarray}
E[f](x)=i\int_{\tilde\Sigma}d\mu(\alpha)\left[\langle\bar
T_\alpha\bar\zeta_\alpha,f\rangle_MT_\alpha(t)\zeta_\alpha(\vec
x)-\langle T_\alpha\bar\zeta_\alpha,f\rangle_M\bar
T_\alpha(t)\zeta_\alpha(\vec x)\right].\label{PropModeDecomp}
\end{eqnarray}

Now we proceed to the quantization. The mode decomposition of
arbitrary CCR quantum fields is discussed in \cite{ZhAThesis}
which provides a generalization of the works by
\cite{Lueders_Roberts_1990}, \cite{Fulling1989},
\cite{Parker1969},\cite{Birrell1984}. Here we summarize some
results applied to the quantized Klein-Gordon field. The latter is
given by the field algebra $\mathcal{A}$ generated by the unit
$\id$ and the elements $\phi(f)$ satisfying
\begin{enumerate}
\item $\phi(af+h)=a\phi(f)+\phi(h)$,

\item $\phi(\bar f)=\phi^*$,

\item $[\phi(f),\phi(h)]=-i\cdot E(f,h)\id$,

\item $\phi((\Box+m^2)f)=0$, \qquad$\forall
f,h\in\mathcal{D}(M)\mbox{, }a\in\mathbb{C}$.
\end{enumerate}
A state $\omega$ of the field is a linear functional
$\omega\in\mathcal{A}'$ such that $\omega(\id)=1$ and
$\omega(A^*A)\ge0$ for all $A\in\mathcal{A}$. The 2-point function
$\omega_2$ of a state $\omega$ is the bilinear form
$\omega_2(f,h)=\omega(\phi(f)\phi(h))$. It follows that
$\omega_2((\Box+m^2)f,h)=\omega_2(f,(\Box+m^2)h)=0$,
$\omega_2(\bar f,f)\ge0$ and
$\omega_2(f,h)-\omega_2(h,f)=-iE(f,h)$. Moreover,
$$
\overline{\omega_2(\bar f,\bar
h)}=\overline{\omega(\phi(f)^*\phi(h)^*)}=\overline{\omega([\phi(h)\phi(f)]^*)}=\omega(\phi(h)\phi(f))=\omega_2(h,f),
$$
$\omega_2$ is hermitian. A quasifree state $\omega$ is a state
which is completely determined by its 2-point function $\omega_2$
(for precise definitions see, e.g.,
\cite{Kay1991},\cite{Araki_Yamagami_1982}). Being a weak
bi-solution of the field equation $\omega_2$ can be mode
decomposed and that is done by {\bf Proposition 4.1} of
\cite{ZhAThesis}. (For earlier results using mode decomposition of
2-point functions see
\cite{Lueders_Roberts_1990},\cite{Olbermann2007},\cite{Degner_Verch_2010}.)
If we denote for convenience
$$
\tilde f^u(\alpha)=\langle T_\alpha\zeta_\alpha,f\rangle_M\mbox{,
}\tilde f^v(\alpha)=\langle\bar T_\alpha\zeta_\alpha,f\rangle_M,
$$
where
$$
\langle f,h\rangle_M=\int_Mdx f(x)h(x),
$$
then a 2-point function can be written as
\begin{eqnarray}
\omega_2(f,h)=a^\omega(\tilde f^u,\tilde
h^u)+\overline{a^\omega(\tilde{\bar{f}}^u,\tilde{\bar{h}}^u)}+b^\omega(\tilde
f^u,\tilde
h^v)+\overline{b^\omega(\tilde{\bar{f}}^u,\tilde{\bar{h}}^v)}+\delta(\tilde
f^u,\tilde h^v),\label{2PointFuncModeDecomp}
\end{eqnarray}
where $a^\omega$ and $b^\omega$ are bi-distributions satisfying
certain symmetry and positivity conditions, and $\delta$ is the
bi-distribution given by the integral kernel of the usual delta
function. But it differs from the delta function because we
identify kernels with bi-distribution using the pairing
$$
a(\tilde f,\tilde
h)=\int_{\tilde\Sigma}\int_{\tilde\Sigma}d\mu(\alpha)d\mu(\beta)a(\alpha,\beta)\tilde
f(\alpha)\tilde h(-\beta).
$$
This strange convention is chosen only for calculational purposes
and can be transformed to the usual form if needed. In particular,
the $\delta$ term in the formula for the 2-point function is
$$
\delta(\tilde f^u,\tilde
h^v)=\int_{\tilde\Sigma}d\mu(\alpha)\tilde f(\alpha)\tilde
h(-\beta).
$$
This term by itself represents a quasifree pure state, and the
remaining part of the generic 2-point function is symmetric. This
state depends on the choice of the modes $T_\alpha$. Choosing
different modes $S_\alpha$ we will find a rich supply of such pure
states. Then we can transform these states back to our original
$T_\alpha$ as follows. Let the old and new modes be related by
$S_\alpha=\mu_\alpha T_\alpha+\nu_\alpha\bar T_\alpha$ with
$|\mu_\alpha|^2-|\nu_\alpha|^2=1$. The pure state given by the
$\delta$ term in modes $S_\alpha$ is determined by the choice
$a^\omega=b^\omega=0$. As described in \cite{ZhAThesis} in the
original modes these components will become
$a^\omega(\alpha,\beta)=\delta(\alpha,\beta)\mu_\alpha\bar\nu_\alpha$
and $b^\omega(\alpha,\beta)=\delta(\alpha,\beta)|\nu_\alpha|^2$.
But such states do not exhaust all pure quasifree states. By {\bf
Corollary 4.1} of \cite{ZhAThesis} any pure quasifree state is
given by $a^\omega(\alpha,\beta)=-\delta(\alpha,\beta)\circ\tilde
S^{v,u}$ and
$b^\omega(\alpha,\beta)=-\delta(\alpha,\beta)\circ\tilde S^{v,v}$,
where the linear maps $\tilde S^{v,u},\tilde S^{v,v}$ are subject
to certain conditions. In fact, the pure states given by $\delta$
terms are special in that they are homogeneous states, i.e., they
are invariant under the isometry group $G$. More generally, by
{\bf Proposition 4.3} of \cite{ZhAThesis} any homogeneous
quasifree state is given by coefficients $a^\omega$, $b^\omega$
which are of the form
\begin{eqnarray}
a^\omega(\tilde f,\tilde
h)=\mathfrak{a}^\omega\left(\int_\mathbb{R}dr\tilde
f(-k,\lambda,r,s)\tilde h(k,\lambda',r,s')\right),\nonumber\\
b^\omega(\tilde f,\tilde
h)=\mathfrak{b}^\omega\left(\int_\mathbb{R}dr\tilde
f(-k,\lambda,r,s)\tilde
h(k,\lambda',r,s')\right),\label{HomStateModeCoeff}
\end{eqnarray}
with distributions $\mathfrak{a}^\omega(k,\lambda,\lambda',s,s')$
and $\mathfrak{b}^\omega(k,\lambda,\lambda',s,s')$.

Another important notion is the notion of Hadamard states. A
quasifree state $\omega$ is said to be Hadamard if its 2-point
function $\omega_2$ satisfies the microlocal spectral condition
($\mu SC$) (see \cite{Radzikowski1996}). Hadamard states are
believed to be the states of physical importance for several
reasons discussed in the literature
(\cite{Kay1991},\cite{Wald1984},\cite{Fulling1989},\cite{Fewster2003}).
One way of checking whether a given $\omega_2$ satisfies the $\mu
SC$ is to try to compute its wave front set directly
\cite{Radzikowski1996}. The mode decomposition suggests another
way of doing this. By {\bf Proposition 4.6} of \cite{ZhAThesis}
there exists a wide variety of modes $T_\alpha$ such that the
homogeneous pure states given by $\delta$ terms in these modes are
Hadamard. Such modes can be computed easily. For some conceptually
related approaches (however emphasizing somewhat different
aspects), see also
\cite{Lueders_Roberts_1990},\cite{Fewster:2003ey},\cite{Junker2002},\cite{Olbermann2007},\cite{Degner_Verch_2010}.
For convenience we switch to the variable
$$
s(t)=\int_0^t\frac{d\tau}{a^3(\tau)}
$$
as described in \cite{ZhAThesis}, then the mode equations become
$$
\ddot T_\alpha(s)+\Lambda_\alpha(s)T_\alpha(s)=0,
$$
where
$$
\Lambda_\alpha(s)=a^6(t(s))(\frac{\lambda}{a^2(t(s))}+m^2),
$$
and $t(s)$ is the inverse function of $s(t)$. Now choose an
arbitrary non-negative function $\eta\in C^\infty_0[-1,0]$ and let
$\tilde\eta(s)=\int_0^sd\sigma\eta(\sigma)$. Denote
$$
\Lambda_\alpha'(s)=(1-\tilde\eta(s))\Lambda_\alpha(-1)+\tilde\eta(s)\Lambda_\alpha(s),
$$
and let
$$
\ddot T_\alpha'(s)+\Lambda_\alpha'(s)T_\alpha'(s)=0.
$$
Choose the initial data $T_\alpha'(-1)=p(\Lambda_\alpha(-1))$ and
$T_\alpha'(0)=q(\Lambda_\alpha(-1))$ with corresponding functions
$p$, $q$ as before with the additional constraint that
$$
p(\Lambda)-\frac{1}{\sqrt[4]{4\Lambda}}=\mathfrak{o}(\Lambda^{-\infty})\mbox{,
}q(\Lambda)-i\sqrt[4]{\frac{\Lambda}{4}}=\mathfrak{o}(\Lambda^{-\infty}),
$$
where $=\mathfrak{o}(\Lambda^{-\infty})$ means of rapid decay at
$\Lambda\to+\infty$. Then the choice $T_\alpha(0)=T_\alpha'(0)$
and $\dot T_\alpha(0)=\dot T_\alpha'(0)$ yields modes $T_\alpha$
such that the homogeneous pure state given by the $\delta$ term is
Hadamard. This is basically the essence of {\bf Proposition 4.6}.
Some relations of the rapid decay in the Fourier space with the
Hadamard property can be found in \cite{Olbermann2007}. It follows
that if such modes $T_\alpha$ are chosen, then a generic quasifree
state is Hadamard if and only if the remaining symmetric part of
$\omega_2$ given by coefficients $a^\omega$ and $b^\omega$ is
smooth. A general criterion for smoothness of a distribution in
terms of it Fourier transform is unfortunately not known in
harmonic analysis. In principle {\bf Proposition 4.4} of
\cite{ZhAPub12012} could give at least sufficient conditions, but
this is still a work in progress.

\section{Conclusion}

In the introduction we have briefly commented on the relevance of
anisotropic (and, in particular, Bianchi) models in cosmology and
mathematics in general terms, which serves as a motivation for
studying them. Now we want to be more specific in what has been
obtained and what can be done with it. Starting with the interests
of pure mathematics we state that in this paper for the first time
(to our knowledge) harmonic analysis in Bianchi IV-VII groups has
been performed including the construction of the dual spaces, the
Plancherel measure and the automorphism groups with their dual
actions. Furthermore, the spectral theory of any metric Laplace
operator is given explicitly for Bianchi II-VII groups, which is
another new result pertaining to spectral analysis. Coming to
mathematical physics, and to quantum field theory in curved
spacetimes in particular, one is interested in rigorous
mathematical description of quantum effects in the presence of
strong classical gravity, which are inherent to the early epoch of
the universe. In \cite{ZhAPub12012} we have analyzed the most
general setup of a linear field in a cosmological model, and we
refer to there for a discussion of the problems and current state
of art in the field. There we have developed a mode decomposition
technique for arbitrary fields in abstract terms. The results are,
however, rather technical and can hardly be used in numerical
calculations for cosmological purposes. To produce some formulae
which can be used in cosmology (as it is done with easier FRW
models) one needs to provide very explicit implementations of the
mathematical structures involved. This has been done in the
present paper, and although the derivations need not be easily
accessible to non-mathematical public, the final formulas indeed
are. One such end-product is the explicit Plancherel measure
$d\nu(k)$ given by (\ref{dnu_k}) along with (\ref{dotnu_k}) from
Section \ref{Section_Plancherel} and formulae for $\mathfrak{K}$
from Section \ref{Section_Orbits}. Another final result is the
spectrum of the Laplace operator, $\Spec(\Delta)=(-\infty,R_h]$,
where $R_h$ is given by (\ref{R_h}) and (\ref{S_Mh}), and the
system of eigenfunctions given by (\ref{Eigenfunc}). This all is
then used to perform the Fourier transform as in Section
\ref{Section_Fourier} and its inverse transform by
(\ref{Fourier_Inv}). We furthermore applied some of the results
for the scalar Klein-Gordon field on Bianchi I-VII spacetimes with
isotropic expansion rate in Section \ref{Section_KG} to enhance
the understanding of their ease and usefulness. Formula
(\ref{lambda_alpha_t}) gives the evolution of the spectrum of the
Schr\"odinger operator of the field. Then the field is mode
decomposed in (\ref{ModeDecompKG}), and the ensuing mode equations
(generalized harmonic oscillators) are given in (\ref{ModeEqKG})
with (\ref{ModeEqCoeff}). The propagator of the field is given
explicitly in terms of modes in (\ref{PropModeDecomp}). Next the
quantized KG field is considered, and the 2-point function of a
quasifree state is mode decomposed in
(\ref{2PointFuncModeDecomp}). It is shown in
(\ref{HomStateModeCoeff}) how the mode coefficients of a state get
simplified if the state is homogeneous, i.e, is invariant under
the Bianchi group of isometries. These results can be used to
extend several results obtained earlier for the KG field on FRW
spacetimes to Bianchi spacetimes (see references cited in Section
\ref{Section_KG}), including the very interesting calculations of
the cosmological particle creation in states of low energy
\cite{Degner_Verch_2010}.

\ack

The first named author thanks the Max Planck Institute for
Mathematics in the Sciences and the International Max Planck
Research School for hosting and financially supporting his PhD
project which this work is a part of. The authors are further
indebted to Professor Gerald Folland for very helpful remarks and
discussions.

\appendix

\section{Induced representations}

Although the generalities of induced representations can be found
in any standard textbook on group representations, for consistency
we will very briefly give an overview of them in this appendix. We
will follow the {\bf Chapter 6} of \cite{Folland199502} in our
treatment.

The inducing procedure produces unitary representations of a
locally compact topological group $G$ from a unitary
representation of its closed subgroup $H$. In this work we will be
interested only in Lie groups, so that the majority of functional
analytical issues are automatically settled. If $\rho$ is a
unitary representation of $H$, and $\nu$ the unitary
representation of $G$ induced from $H$ (described below), we will
write $\nu=\mbox{\bf Ind}_H^G\rho$. In particular the restriction
of $\nu$ to $H$ is unitarily equivalent to $\rho$,
$\nu|_H\sim\rho$. Unless the homogeneous space $G/H$ is a finite
set, $\nu$ is infinite dimensional.

Denote $M=G/H$. We present the construction of the induced
representation under assumptions which hold in cases of our
interest. Namely, we suppose that there exists a $G$-invariant
measure $dx_M$ on $M$. Any $x\in G$ can be uniquely written as
$x=x_Mx_H$ with $x_M\in M$ and $x_H\in H$. If $\mathcal{H}_\rho$
is the representation Hilbert space of $\rho$, then the
representation Hilbert space $\mathcal{H}_\nu$ of the induced
representation is taken to be
$\mathcal{H}_\nu=L^2(M,\mathcal{H}_\rho,dx_M)$, i.e.,
$\mathcal{H}_\rho$-valued $dx_M$-square integrable functions on
$M$. The action of the representation $\nu$ on $\mathcal{H}_\nu$
is given by
$$
\nu(x)f(y)=\rho((x^{-1}y)_H)^{-1}f((x^{-1}y)_M)\mbox{, }\forall
x\in G\mbox{, }y\in M\mbox{, }f\in\mathcal{H}_\nu.
$$
That this is a natural construction can be seen by the following
nice properties. If $\rho$ and $\rho'$ are unitarily equivalent
unitary representations of $H$, then $\nu=\mbox{\bf Ind}_H^G\rho$
and $\nu'=\mbox{\bf Ind}_H^G\rho'$ are unitarily equivalent
unitary representations of $G$. Moreover, it can be shown that if
$\{\rho_i\}$ is a family of unitary representations of $H$ then
$\mbox{\bf Ind}_H^G\bigoplus\rho_i$ is unitarily equivalent to
$\bigoplus\mbox{\bf Ind}_H^G\rho_i$.

\section{An illustration of co-adjoint orbits}

\includegraphics[scale=0.5]{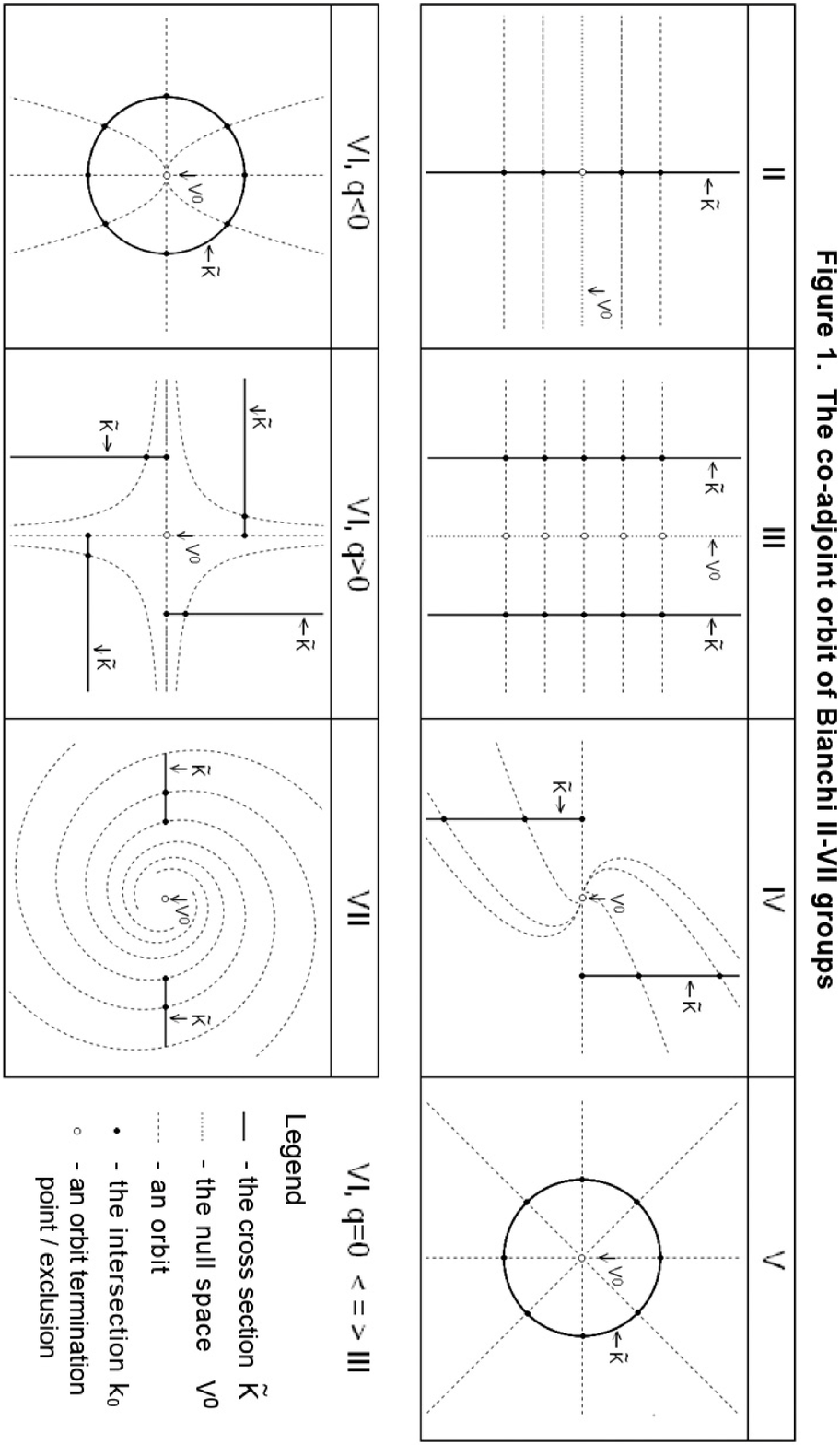}

\section*{Bibliography}

\bibliographystyle{ieeetr}
\bibliography{lib}

\end{document}